\documentclass[12pt]{article}

\usepackage{amsthm,amsmath,amssymb}
\usepackage{hyperref}
\usepackage{enumerate}
\usepackage{comment}
\usepackage[misc,geometry]{ifsym}
\usepackage{bm}
\usepackage{mathrsfs}
\usepackage{makeidx}

\newtheorem{thm}{Theorem}[section]
\newtheorem{co}[thm]{Corollary}
\newtheorem{lem}[thm]{Lemma}

\newtheorem{assumption}[thm]{Assumption}

\newtheorem{definition}[thm]{Definition}
\newenvironment{de}{\begin{definition}\rm}{\end{definition}}
\newtheorem{example}[thm]{Example}
\newenvironment{exmp}{\begin{example}\rm}{\end{example}}
\newtheorem{remark}[thm]{Remark}
\newenvironment{rem}{\begin{remark}\rm}{\end{remark}}

\makeindex

\title{An Optimization Approach to the Langberg-M\'{e}dard Multiple Unicast Conjecture~\footnote{A preliminary version~\cite{CH18} of this work has been presented in IEEE ISIT 2018.}~\thanks{This research is partly supported by a grant from the Research Grants Council of the Hong Kong Special Administrative Region, China (Project No. 17301017).}
}

\author{\small \begin{tabular}{ccc}
Kai Cai & Guangyue Han\\
The University of Hong Kong& The University of Hong Kong\\
email: kcai@hku.hk & email: ghan@hku.hk\\
\end{tabular}}

\date{{\normalsize \today}}

\begin{document} \maketitle

\begin{abstract}
The Langberg-M\'{e}dard multiple unicast conjecture claims that for any strongly reachable $k$-pair network, there exists a multi-flow with rate $(1,1,\dots,1)$. In a previous work, through combining and concatenating the so-called elementary flows, we have constructed a multi-flow with rate at least $(\frac{8}{9}, \frac{8}{9}, \dots, \frac{8}{9})$ for any $k$. In this paper, we examine an optimization problem arising from this construction framework. We first show that our previous construction yields a sequence of asymptotically optimal solutions to the aforementioned optimization problem. And furthermore, based on this solution sequence, we propose a perturbation framework, which not only promises a better solution for any $k \mod 4 \neq 2$ but also solves the optimization problem for the cases $k=3, 4, \dots, 10$, accordingly yielding multi-flows with the largest rate to date.
\end{abstract}

\section{Introduction} \label{Introduction}

We consider a directed $k$-pair network $\mathcal N=(V, A, S, R)$, which consists of an underlying digraph $D=(V, A)$, $k$ senders $S=\{s_1, s_2, \dots, s_k\}\subseteq V$ and $k$ receivers $R=\{r_1, r_2, \dots, r_k\}\subseteq V$. Let $\overline{\mathcal{N}}$ denote the underlying undirected network of $\mathcal{N}$, where the orientation in $\mathcal{N}$ is ignored. Throughout this paper, we assume that each arc in $\mathcal{N}$ (and as a result, each edge in $\overline{\mathcal{N}}$) is of unit capacity. The network coding rate $\mathbf{R}_{c}(\mathcal{N})$ is a real vector $(d_1, d_2, \dots, d_k)$ such that $d_i$ is the transmission rate from $s_i$ to $r_i$ when using network coding along the orientation of $\mathcal{N}$, while the {\em routing rate} $\mathbf{R}_{c}(\overline{\mathcal{N}})$ is a real vector $(d_1, d_2, \dots, d_k)$ such that there exists a feasible $(s_1, s_2, \dots,s_k)$-$(r_1, r_2, \dots, r_k)$ multi-flow (see definition in Section \ref{subsection-multi-flow-basics}) over $\mathcal{N}$.

One of the most fundamental problems in the theory of network coding is {\em the multiple unicast network coding conjecture}~\cite{Li042}, or simply {\em the multiple unicast conjecture}, which states that for any $\overline{\mathcal{N}}$, the transmission rate achieved by any fractional network coding can be achieved by routing as well. Despite enhanced understanding in certain special cases, the conjecture has been doggedly resisting a series of attacks~\cite{Harv06}-\cite{Langberg09} and turned out to be one of the most hardest problems in network coding theory.

A weaker version of the conjecture, proposed by Langberg and M\'{e}dard \cite{Langberg09}, focuses on a strongly reachable $k$-pair network $\mathcal{N}$ and claims that (the ``$\leq$'' below should be interpreted in the pairwise sense)
$$
\mathbf{R}_{c}(\mathcal N)\leq \mathbf{R}_{r}(\overline{\mathcal{N}}),
$$
Here, a $k$-pair network $\mathcal{N}$ is said to be {\em strongly reachable} if there exists an $s_i$-$r_j$ directed path $P_{s_i, r_j}$ for all feasible $i, j$, and the paths $P_{s_1, r_j}, P_{s_2, r_j}, \cdots, P_{s_k, r_j}$ are {\em edge-disjoint} for each feasible $j$. Apparently, for a strongly reachable $k$-pair network $\mathcal N$, $\mathbf{R}_c(\mathcal N)\geq(1,1,\dots,1)$. So, if the multiple unicast conjecture is true, one will deduce that $\mathbf{R}_r(\overline{\mathcal{N}}) \geq (1,1,\dots,1)$. Langberg and M\'{e}dard~\cite{Langberg09} showed that $\mathbf{R}_r(\overline{\mathcal{N}}) \geq (\frac{1}{3}, \frac{1}{3}, \dots, \frac{1}{3})$, which was further improved to $\mathbf{R}_r(\overline{\mathcal{N}}) \geq (\frac{8}{9}, \frac{8}{9}, \dots, \frac{8}{9})$ in~\cite{CH15} by way of combining and concatenating some so-called elementary flows.

In this paper, we will examine a sequence of optimization problems $\{\mathcal{P}_{\mathcal{S}_k}\}$, whose optimal solutions will naturally give lower bounds on $\mathbf{R}_r(\overline{N})$. We first prove that our construction in~\cite{CH15} yields $\{\mathcal{C}_k^*\}$, a sequence of asymptotically optimal solutions to $\{\mathcal{P}_{\mathcal{S}_k}\}$. And furthermore, based on $\{\mathcal{C}_k^*\}$, we propose a perturbation framework to obtain $\{\mathcal{C}_k^{**}\}$, which promises a better solution than $\mathcal{C}_k^*$ for any $k \mod 4 \neq 2$ and further solves $\mathcal{P}_{\mathcal{S}_k}$ for $k=3, 4, \dots, 10$, and thereby yielding multi-flows with the largest rate to date (see Section~\ref{routing-rate}). Here we note that a prototypical version of the optimization problem $\mathcal{P}_{\mathcal{S}_k}$ was first proposed in~\cite{CH17}. The progress made in this work is due to a (rather) delicate study of structural and analytic aspects of the optimization problem, which include symmetries, asymptotics, behaviors upon perturbation and so on.

The rest of paper is organized as follows. In Section~\ref{section-the-basic}, we give some basic notions and facts in the theory of multi-flows, and we introduce the optimization problem $\mathcal{P}_{\mathcal{S}_k}$ and elaborate its connections with the theory of multi-flows. In Section~\ref{section-sym-asym}, we investigate the symmetries of the optimization problem $\mathcal{P}_{\mathcal{S}_k}$ and give the limit of its optimal value as $k$ tends to infinity. We introduce the so-called strong homogeneous flow $\mathcal C^*_k$ in Section~\ref{section-shf}, where we first show that $\{\mathcal C^*_k\}$ is a sequence of asymptotically optimal solution for $\{\mathcal P_{\mathcal S_k}\}$ and then prove that it gives the exact optimal solution when $k=1, 2, 6, 10$. In Section~\ref{section-perturbation}, we propose a unified framework to perturb $\mathcal C^*_k$ to obtain a better solution $\mathcal C^{**}_k$ for any $k \mod 4 \neq 2$. In Sections~\ref{section-perturb}, we give $\mathcal C^{**}_k$ for $k=3, 4, 5, 7, 8, 9$ explicitly, and we further establish the optimality of these $\mathcal C^{**}_k$ and their uniqueness for achieving optimality. Finally, the paper is concluded in Section~\ref{section-conclusion}.

\section{Mathematical Preliminaries}\label{section-the-basic}

\subsection{Multi-Flow Basics}\label{subsection-multi-flow-basics}

Let $D=(V, A)$ be a directed graph with vertex set $V$ and arc set $A$. For an arc $a=(u,v)\in A$, let $tail(a)$, $head(a)$ denote its $tail$ $u$, $head$ $v$, respectively. For any $s, r \in V$, an $s$-$r$ {\em flow} is a function $f: A \rightarrow \mathbb{R}$ satisfying the following {\em flow conservation law:} for any $v \notin \{s, r\}$,
\begin{equation}\label{flow conservation law}
excess_f(v)=0,
\end{equation}
where
\begin{equation}
excess_{f}(v):=\sum_{a\in A: \; head(a)=v} f(a)-\sum_{a\in A: \; tail(a)=v} f(a).
\end{equation}
It is easy to see that $|excess_{f}(s)|=|excess_{f}(r)|$, which is called the {\em value (or rate)} of $f$. Note that the above definitions naturally give rise to a {\em fractional flow} on the underlying undirected graph of $D$, and it is not needed to differentiate an $s$-$r$ flow from an $r$-$s$ flow. This is different from Schrijver~\cite{Schrijver03}, where a flow must be a non-negative function.

There are two kinds of operations on the flows defined as above. Firstly, the set of all $s$-$r$ flows naturally forms a linear space over $\mathbb R$; particularly, for any two $s$-$r$ flows $f_1, f_2$ and scalars $u, v\in \mathbb R$, and the function $f=u f_1+ v f_2$ is again an $s$-$r$ flow. Secondly, let $f$ be an $s$-$t$ flow and $g$ be a $t$-$r$ flow such that
$$
excess_f(t)=-excess_g(t).
$$
Then by definition, $f+g$ is an $s$-$r$ flow, which is called the {\em concatenation} of $f$ and $g$. Adopting the notational convention in defining the concatenation of paths in~\cite{Schrijver03}, the concatenation of $f$ and $g$ will be denoted by $f g$.

An $(s_1, s_2, \dots,s_k)$-$(r_1, r_2, \dots, r_k)$ {\em multi-flow} refers to a set of $k$ flows  $\mathcal{F}=\{f_{i}: i=1, 2, \dots, k\}$, where each $f_{i}$ is an $s_i$-$r_i$ flow. We say $\mathcal{F}$ has {\em rate} $(d_1, d_2, \dots, d_k)$, where $d_i:=|excess_{f_i}(s_i)|$; and, for any given $a \in A$, we define $|\mathcal{F}|(a)$ as
\begin{equation}\label{total value in an arc}
|\mathcal{F}|(a):=\sum_{1\leq i\leq k}|f_{i}(a)|.
\end{equation}
The multi-flow $\mathcal{F}=\{f_{i}: i=1, 2, \dots, k\}$ is said to be {\em feasible} with respect to capacity function $c$ if $|\mathcal{F}|(a) \leq c(a)$ for all $a\in A$. Note that when $k=1$, the multi-flow is just a flow $f$, and $f$ is feasible if $|f(a)|\leq c(a)$ for all $a\in A$ (Here recall that we have assumed $c(a)\equiv 1$ in Section~\ref{Introduction}).

\subsection{Elementary Flows}\label{subsection-flow construction}

For a {\em strongly reachable} $k$-pair network $\mathcal{N}=(V, A, S, R)$, let $\mathbf{P}=\{P_{s_i, r_j}\}_{i,j=1}^k$ be a set of $s_i$-$r_j$ directed paths, where the paths $P_{s_1, r_j}, P_{s_2, r_j}, \cdots, P_{s_k, r_j}$ are {\em edge-disjoint} for each feasible $j$. For each $P_{s_i, r_j}\in \mathbf{P}$, define an $s_i$-$r_j$ flow as follows:
$$
f_{i,j}(a)=\left\{
              \begin{array}{ll}
                1, & \hbox{$a\in P_{s_i, r_j}$,}\\
                0, & \hbox{otherwise.}
              \end{array} \right.
$$
Let $\mathbf{F}=\{f_{i,j}| 1\leq i,j\leq k\}$, a set of {\em elementary flows} with respect to $\mathbf{P}$, which will be the ``building blocks'' for the multi-flow construction in this paper.

More specifically, let
$$
\mathcal C=\left((c^{(1)}_{i,j}), (c^{(2)}_{i,j}), \dots, (c^{(k)}_{i,j})\right)
$$
be a $k$-tuple of $k\times k$ real matrices. And for $\ell=1, 2, \dots, k$, consider $\mathcal F=\{f_1, f_2, \cdots, f_k\}$,
where
\begin{equation} \label{coefficient-matrix}
f_{\ell}=\sum_{i,j=1}^k c^{(\ell)}_{i,j} f_{i,j}.
\end{equation}
The following theorem says that if $\mathcal{C}$ satisfies certain conditions, then the constructed $\mathcal F$ in (\ref{coefficient-matrix}) is also a multi-flow. 
\begin{thm}\label{basic observation}
$\mathcal F=\{f_1, f_2, \cdots, f_k\}$ be an $(s_1, s_2, \dots, s_k)$-$(r_1,r_2,\dots, r_k)$ multi-flow with rate $(1,1,\dots,1)$ if and only if each $(c^{(\ell)}_{i,j})$ satisfies:
\begin{equation}\label{commodity condition}
\begin{split}
    &1)\; \sum_{j=1}^kc^{(\ell)}_{i,j}=0, \text{ for all}\; i\neq \ell;\\
    &2)\; \sum_{i=1}^kc^{(\ell)}_{i,j}=0, \text{ for all}\; j\neq \ell;\\
    &3)\;\sum_{i=1}^k\sum_{j=1}^kc_{i,j}^{(\ell)}\equiv 1.
\end{split}
\end{equation}
\end{thm}

\begin{proof}
We only need to prove that $f_{\ell}$ is an $s_{\ell}$-$r_{\ell}$ flow with rate $1$ for any $\ell$. To see this,
Note that $excess_{f_{\ell}}(s_i)=\sum_{j=1}^k c^{(\ell)}_{i,j}$ and $excess_{f_{\ell}}(r_j)=\sum_{i=1}^k c^{(\ell)}_{i,j}$. Condition $1)$ implies that the conservation law is satisfied by all the senders except $s_{i_0}$; Condition $2)$ implies that it is satisfied by all the receivers except $r_{j_0}$; Condition $3)$ implies that the value of $f_{\ell}$ is $1$.
\end{proof}

\subsection{The Optimization Problem $\mathcal{P}_{\mathcal{S}_k}$} \label{subsection-the-opt-problem}

The optimization problem $\mathcal{P}_{\mathcal{S}_k}$ to be introduced in this section is intimately connected with our multi-flow construction and will be the main subject of study in this paper.

Let
$$
\mathfrak S_k:=\left\{\mathcal C=\left((c^{(1)}_{i,j}), (c^{(2)}_{i,j}), \dots, (c^{(k)}_{i,j})\right)\| \;\mathcal C \;\text{satisfies (\ref{commodity condition})}\right\}.
$$
Clearly, $\mathfrak S_k$ is defined by a total of $2k-1$ linearly independent constraints and is an affine subspace of $\mathbb R^{k^3}$ with dimension $k(k-1)^2$.

Throughout this paper, we will refer to a non-empty subset of $[k] \times [k]$ as a {\em $k$-sample}, where $[k]:=\{1,2,\dots,k\}$.
For a given $k$-sample $s$ and $\ell \in [k]$, we define a function $g^{(\ell)}_{s}: \mathfrak S_k \rightarrow \mathbb R$ as
$$
g^{(\ell)}_{s}(\mathcal C):=\underset{(i,j)\in s}{\sum}c^{(\ell)}_{i,j},
$$
based on which, we define $g_{s}: \mathfrak S_k \rightarrow \mathbb R$ as
$$
g_{s}(\mathcal C):=\sum_{\ell=1}^k|g^{(\ell)}_{s}(\mathcal C)|.
$$
Furthermore, for a non-empty set $\mathcal{S}$ of $k$-samples, we define
$$
g_{\mathcal{S}}(\mathcal C):=\underset{s\in \mathcal{S}}{\text{max}}\{g_{s}(\mathcal C)\}.
$$

Now, we are ready to introduce the optimization problem $\mathcal{P}_{\mathcal{S}}$ as follows:
\begin{equation}  \tag{$\mathcal{P}_{\mathcal{S}}$}
\begin{aligned}
& \text{minimize} & & g_{\mathcal{S}}(\mathcal C)\\
& \text{subject to} & & \mathcal C\in \mathfrak S_k.
\end{aligned}
\end{equation}
Note that $g_{\mathcal{S}}$ is continuous and lower bounded, and thereby its optimal value is achievable, i.e., there exists an {\em optimal solution (optimal point)} $\bar{\mathcal{C}} \in \mathfrak S_k$ such that
$$
g_{\mathcal{S}}(\bar{\mathcal{C}})=\underset{\mathcal C\in \mathfrak S_k}{\min}g_{\mathcal{S}}(\mathcal C).
$$

The following theorem says that $\mathcal{P}_{\mathcal{S}}$ is a convex optimization problem. Though the theorem follows from a standard argument, we give its proof for the sake of completeness.
\begin{lem} \label{convex}
$g_{\mathcal{S}}$ is a convex function over $\mathfrak S_k$.
\end{lem}

\begin{proof}
Let $\mathcal C_1, \mathcal C_2\in \mathfrak S_k$ and $\mathcal C=p\, \mathcal C_1+q\, \mathcal C_2$, such that $p+q=1$, $p,q\leq 0$. Clearly, $\mathcal C\in \mathfrak S_k$, and
\begin{equation*}
\begin{split}
g_{\mathcal{S}}(\mathcal C)&=\max_{s\in S}\{g_s(\mathcal C)\}\\
&=\max_{s\in \mathcal{S}}\{|g^{(1)}_{s}(p\, \mathcal C_1+q\, \mathcal C_2)|+\dots+|g^{(k)}_{s}(p\, \mathcal C_1+q\, \mathcal C_2)|\}\\
&\leq \max_{s\in \mathcal{S}}\{p|g^{(1)}_{s}(\mathcal C_1)|+q|g^{(1)}_{s}(\mathcal C_2)|+\dots+p|g^{(k)}_{s}(\mathcal C_1)|+q|g^{(k)}_{s}(\mathcal C_2)|\}\\
&=\max_{s\in \mathcal{S}}\{ p\, g_s(\mathcal C_1)+q\, g_s(\mathcal C_2)\}\\
&\leq p\,\max_{s\in \mathcal{S}}\{g_s(\mathcal C_1)\}+q\,\max_{s\in S}\{g_s(\mathcal C_2)\}\\
&=p\,g_{\mathcal{S}}(\mathcal C_1)+q\,g_{\mathcal{S}}(\mathcal C_2),
\end{split}
\end{equation*}
which completes the proof.
\end{proof}

\begin{de}[{\bf Strongly Reachable Sample Set}]
The following set of $k$-samples, denoted by $\mathcal{S}_k$, is of particular interest for the consideration of strongly reachable $k$-pair networks:
\begin{equation} \label{S-k}
\mathcal{S}_k := \{\{(i_1,j_1),\dots (i_r,j_r)\}\subseteq [k]\times[k] \| \, j_1 < j_2 < \dots < j_r, 1 \leq r \leq k\}.
\end{equation}
Put it differently, $\mathcal{S}_k$ is composed of all the $k$-samples, each of which consists of elements whose two coordinates are distinct. Clearly, there are $(k+1)^k-1$ samples in $\mathcal{S}_k$.
\end{de}

\begin{exmp}\label{ex-s-sample}
It is easy to see that $\mathcal{S}_1=\{\{(1, 1)\}\}$. And $\mathcal{S}_2$ is composed of $8$ samples, $\{(1,1)\}$, $\{(2,1)\}$, $\{(1,2)\}$, $\{(2,2)\}$, $\{(1,1),(1,2)\}$, $\{(2,1),(1,2)\}$, $\{(1,1),(2,2)\}$, and $\{(2,1),(2,2)\}$. And one can verify that $\mathcal{S}_3$ contains $63$ samples and $\mathcal{S}_4$ contains $624$ samples.
\end{exmp}

Let $\mathcal{O}_{\mathcal S_k}$ denote the {\em optimal value} of $\mathcal P_{\mathcal S_k}$. The following theorem provide a key link connecting $\mathcal{P}_{\mathcal{S}_k}$ and $\mathcal R_r(\overline{\mathcal N})$, where $\mathcal{N}$ is an arbitrary strongly reachable $k$-pair network.
\begin{thm}\label{basic lemma}
For any strongly reachable $k$-pair network $\mathcal N$,
$$
\mathcal R_r(\overline{\mathcal N}) \geq \left(\frac{1}{\mathcal{O}_{\mathcal S_k}}, \frac{1}{\mathcal{O}_{\mathcal S_k}}, \dots, \frac{1}{\mathcal{O}_{\mathcal S_k}}\right).
$$
\end{thm}

\begin{proof}
Let $\bar{\mathcal{C}}=\left((\bar{c}^{(1)}_{i,j}),(\bar{c}^{(2)}_{i,j}), \dots, (\bar{c}^{(k)}_{i,j})\right)$ be an optimal point for $\mathcal P_{\mathcal S_k}$, that is to say, $\mathcal{O}_{\mathcal S_k}=g_{\mathcal{S}_k}(\bar{\mathcal{C}})$. And let $\mathcal F=\{f_1, f_2, \dots, f_k\}$ be the $(s_1,s_2,\dots,s_k)$-$(r_1,r_2,\dots,r_k)$ multi-flow constructed from $\mathbf F$ with coefficient matrices $\frac{1}{\mathcal{O}_{\mathcal S_k}}\mathcal C^*=\left(\frac{1}{\mathcal{O}_{\mathcal S_k}}(\bar{c}^{(1)}_{i,j}),\frac{1}{\mathcal{O}_{\mathcal S_k}}(\bar{c}^{(2)}_{i,j}),\dots, \frac{1}{\mathcal{O}_{\mathcal S_k}}(\bar{c}^{(k)}_{i,j})\right)$. Clearly, $\mathcal F$ achieves rate $\left(\frac{1}{\mathcal{O}_{\mathcal S_k}}, \frac{1}{\mathcal{O}_{\mathcal S_k}}, \dots, \frac{1}{\mathcal{O}_{\mathcal S_k}}\right)$.

To complete the proof, we only need to prove that $\mathcal F$ is feasible. Towards this goal, for each arc $a$, let $\mathcal P(a)=\{P_{s_{i_1},r_{j_1}}, P_{s_{i_1},r_{j_1}}, \dots, P_{s_{i_\alpha(a)},r_{j_\alpha(a)}}\}\subseteq \mathbf P$ be the set of all the paths passing through $a$. By the definition of a strongly reachable $k$-pair network, we have
$$
s(a):=\{(i_1,j_1),(i_2,j_2),\dots, (i_\alpha(a),j_\alpha(a))\}\in \mathcal{S}_k.
$$
Hence, we have
\begin{equation*}
\begin{split}
|\mathcal F|(a)&=|f_1(a)|+\dots+|f_k(a)|\\
&=\left|f^{(1)}_{s(a)}\left(\frac{1}{\mathcal{O}_{\mathcal S_k}}(\bar{c}^{(k)}_{i,j})\right)\right|
+\dots+\left|f^{(k)}_{s(a)}\left(\frac{1}{\mathcal{O}_{\mathcal S_k}}(\bar{c}^{(k)}_{i,j})\right)\right|\\
&=f_{s(a)}\left(\frac{1}{\mathcal{O}_{\mathcal S_k}} \bar{\mathcal{C}}\right)\\
&\leq 1,
\end{split}
\end{equation*}
which implies that $\mathcal F$ is feasible and thus completes the proof.
\end{proof}

\section{Symmetries and Asymptotics of $\mathcal P_{\mathcal S_k}$} \label{section-sym-asym}

Starting from this section, we will focus on solving the problem $\mathcal P_{\mathcal S_k}$. Apparently, the problem $\mathcal P_{\mathcal S_1}$ is trivial. In~\cite{CH17}, we have shown that $\mathcal P_{\mathcal S_2}$ has optimal value $1$, which is achieved by the unique optimal point $$
\left(\left(
   \begin{array}{cc}
     \frac{3}{4} & \frac{1}{4} \\
     \frac{1}{4} & \frac{-1}{4} \\
   \end{array}
 \right),
\left(
   \begin{array}{cc}
     \frac{-1}{4} & \frac{1}{4} \\
     \frac{1}{4} & \frac{3}{4} \\
   \end{array}
 \right)
\right).
$$
However, the problem $\mathcal P_{\mathcal S_k}$, $k \geq 3$, becomes prohibitively complex and cannot be dealt with a case analysis as in~\cite{CH17}. Rather than a fixed $\mathcal P_{\mathcal S_k}$, this section is devoted to the asymptotics of $\{\mathcal P_{\mathcal S_k}\}$; more precisely we will establish $\lim_{k\rightarrow\infty}\mathcal{O}_{\mathcal S_k}=9/8$, which, as will be shown in Section~\ref{section-shf}, can be achieved by a sequence of explicitly constructed solutions. As elaborated below, the key observation in deriving this result is some symmetric properties possessed by the optimal solutions of $\mathcal P_{\mathcal S_k}$.

\subsection{Symmetries of $\mathcal P_{\mathcal S_k}$}\label{subsection-symmetric}

In this section, we use $Sym(k)$ to denote the symmetric group on $[k]$. Note that a permutation in $Sym(k)$ can be written by a product of disjoint cyclic permutations (cycles), e.g., $\sigma=(15)(342)\in Sym(5)$. For any $\sigma\in Sym(k)$ and $\mathcal C=\left((c^{(1)}_{i,j}), (c^{(2)}_{i,j}), \dots, (c^{(k)}_{i,j})\right) \in \mathfrak S_k$, we define
$$
\sigma(\mathcal{C}):=\widetilde{\mathcal{C}},
$$
where $\widetilde{\mathcal C}=\left((\widetilde{c}^{(1)}_{i,j}), (\widetilde{c}^{(2)}_{i,j}), \dots, (\widetilde{c}^{(k)}_{i,j})\right)$ with $\widetilde{c}^{(\ell)}_{i,j}=c^{(\sigma^{-1}(\ell))}_{\sigma^{-1}(i),\sigma^{-1}(j)}$
for all feasible $i,j,\ell$. Apparently, $\sigma$ defines a one-to-one mapping from $\mathfrak S_k$ to $\mathfrak S_k$.

\begin{exmp}
Let $$\mathcal C=\left(\left(
              \begin{array}{ccc}
                \frac{5}{9} & \frac{2}{9} & \frac{2}{9} \\
                \frac{2}{9} & \frac{-1}{9} & \frac{-1}{9} \\
                \frac{2}{9} & \frac{-1}{9} & \frac{-1}{9} \\
              \end{array}
            \right),
            \left(
              \begin{array}{ccc}
                0 & 0 & 0 \\
                0 & 1 & 0 \\
                0 & 0 & 0 \\
              \end{array}
            \right),
            \left(
              \begin{array}{ccc}
                \frac{-1}{9} & \frac{-1}{9} & \frac{2}{9} \\
                \frac{-1}{9} & \frac{-1}{9} & \frac{2}{9} \\
                \frac{2}{9} & \frac{2}{9} & \frac{5}{9} \\
              \end{array}
            \right)\right)\in \mathfrak S_3,
$$
and let $\sigma_1=(12)$ and $\sigma_2=(123)\in Sym(3)$. Then, we have
$$
\sigma_1(\mathcal C)=\left(\left(
              \begin{array}{ccc}
                1 & 0 & 0 \\
                0 & 0 & 0 \\
                0 & 0 & 0 \\
              \end{array}
            \right), \left(
              \begin{array}{ccc}
                \frac{-1}{9} & \frac{2}{9} & \frac{-1}{9} \\
                \frac{2}{9} & \frac{5}{9} & \frac{2}{9} \\
                \frac{-1}{9} & \frac{2}{9} & \frac{-1}{9} \\
              \end{array}
            \right),
            \left(
              \begin{array}{ccc}
                \frac{-1}{9} & \frac{-1}{9} & \frac{2}{9} \\
                \frac{-1}{9} & \frac{-1}{9} & \frac{2}{9} \\
                \frac{2}{9} & \frac{2}{9} & \frac{5}{9} \\
              \end{array}
            \right)\right),
$$
$$
\sigma_2(\mathcal C)=\left(\left(
              \begin{array}{ccc}
                \frac{5}{9} & \frac{2}{9} & \frac{2}{9} \\
                \frac{2}{9} & \frac{-1}{9} & \frac{-1}{9} \\
                \frac{2}{9} & \frac{-1}{9} & \frac{-1}{9} \\
              \end{array}
            \right),
            \left(
              \begin{array}{ccc}
                \frac{-1}{9} & \frac{2}{9} & \frac{-1}{9} \\
                \frac{2}{9} & \frac{5}{9} & \frac{2}{9} \\
                \frac{-1}{9} & \frac{2}{9} & \frac{-1}{9} \\
              \end{array}
            \right),
            \left(
              \begin{array}{ccc}
                0 & 0 & 0 \\
                0 & 0 & 0 \\
                0 & 0 & 1 \\
              \end{array}
            \right)\right).
$$
\end{exmp}

\begin{de}[{\bf Fixed Point and Invariant Space}] \label{fixed-point-definition}
Let $\mathcal C\in \mathfrak S_k$. $\mathcal C$ is called a {\em fixed point} if for all $\sigma\in Sym(k)$, $\sigma(\mathcal C)=\mathcal C$. The set of all the fixed points is called the {\em invariant space} of $\mathfrak S_k$, and will be denoted by $\mathfrak S_k^{fix}$.
\end{de}

The following theorem shows that $\mathfrak S_k^{fix}$ is in fact a 2-dimensional affine subspace of $\mathfrak S_k$.
\begin{thm}\label{fixedpoint}
Let $\mathcal C=\left((c^{(1)}_{i,j}), (c^{(2)}_{i,j}), \dots, (c^{(k)}_{i,j})\right)\in \mathfrak S_k$.
Then, $\mathcal C\in \mathfrak S_k^{fix}$ if and only if $\mathcal{C}$ takes the following form:
\begin{equation}\label{fixed-form}
\mathcal C=\left(\left(
              \begin{array}{ccccc}
                x & a & a & \dots & a \\
                a & y & b & \dots & b \\
                a & b & y & \dots & b \\
                \vdots & \vdots &  & \ddots & \vdots \\
                a & b & b &\dots & y \\
              \end{array}
            \right), \left(
              \begin{array}{ccccc}
                y & a & b & \dots & b \\
                a & x & a & \dots & a \\
                b & a & y & \dots & b \\
                \vdots & \vdots &  & \ddots & \vdots \\
                b & a & b &\dots & y \\
              \end{array}
            \right),\dots, \left(
              \begin{array}{ccccc}
                y & b & b & \dots & a \\
                b & y & b & \dots & a \\
                b & b & y & \dots & a \\
                \vdots & \vdots &  & \ddots & \vdots \\
                a & a & a &\dots & x \\
              \end{array}
            \right)\right),
\end{equation}
where $x+(k-1)a=1$ and $y+a+(k-2)b=0$.
\end{thm}

\begin{proof}
Clearly, if $\mathcal C$ takes the form in (\ref{fixed-form}), then it is a fixed point. So we only need to prove the reverse direction.

Let $\mathcal{J}$ denote the set of all the entries of $\mathcal C$, i.e., $\mathcal{J}:=\{c^{(\ell)}_{i,j}: 1\leq i,j,\ell\leq k\}$. Consider the group action of $Sym(k)$ on $\mathcal{J}$ with $\sigma(c^{(\ell)}_{i,j})=c^{(\sigma(\ell))}_{\sigma(i),\sigma(j)}$ for any $\sigma\in Sym(k)$ and any $c^{(\ell)}_{i,j}\in \mathcal{J}$. Clearly, under this group action, $\mathcal{J}$ is partitioned into the following orbits: 1) $\mathcal{J}_1:=\{c^{(\ell)}_{i,j}: i=j=\ell\}$; 2) $\mathcal{J}_2:=\{c^{(\ell)}_{i,j}: i=j\neq \ell\}$; 3) $\mathcal{J}_3:=\{c^{(\ell)}_{i,j}: i=\ell\neq j\}$; 4) $\mathcal{J}_4:=\{c^{(\ell)}_{i,j}: j=\ell\neq j\}$; 5) $\mathcal{J}_5:=\{c^{(\ell)}_{i,j}: i\neq j, i\neq\ell, j\neq \ell\}$. It follows from the assumption that $\mathcal C$ is a fixed point that $c^{(\ell)}_{i,j}=c^{(\sigma(\ell))}_{\sigma(i),\sigma(j)}$ for any feasible $i,j,\ell$ and any $\sigma\in Sym(k)$. In other words, the elements in a same orbit must have a same value, and therefore we can assume the existence of $x, y, a_1, a_2, b$ such that
$$
c^{(\ell)}_{i,j}=\left\{
                    \begin{array}{ll}
                      x, & \hbox{if $c^{(\ell)}_{i,j}\in \mathcal{J}_1$,} \\
                      y, & \hbox{if $c^{(\ell)}_{i,j}\in \mathcal{J}_2$,} \\
                      a_1, & \hbox{if $c^{(\ell)}_{i,j}\in \mathcal{J}_3$,} \\
                      a_2, & \hbox{if $c^{(\ell)}_{i,j}\in \mathcal{J}_4$,} \\
                      b, & \hbox{if $c^{(\ell)}_{i,j}\in \mathcal{J}_5$.}
                    \end{array}
                  \right.
$$
Note that from (\ref{commodity condition}), we can deduce that for any $\ell$, $\sum_{j=1}^{k} c^{(\ell)}_{\ell,j}=\sum_{i=1}^{k} c^{(\ell)}_{i,\ell}$, which implies $\sum_{j:j\neq\ell} c^{(\ell)}_{\ell,j}=\sum_{i:i\neq\ell} c^{(\ell)}_{i,\ell}$, or equivalently, $(k-1)a_1=(k-1)a_2$. The proof of the theorem is then complete after writing $a_1, a_2$ as $a$.
\end{proof}

For any $\sigma\in Sym(k)$ and any $k$-sample $s=\{(i_1,j_1),\dots,(i_r,j_r)\}$, we define
$$
\sigma(s):=\{(\sigma(i_1),\sigma(j_1)), \dots, (\sigma(i_r),\sigma(j_r))\}.
$$
It is easy to see that $\sigma$ defines a one-to-one mapping from $2^{[k]\times[k]}$ to $2^{[k]\times[k]}$. For a quick example, let $s=\{(2,1), (1,2), (3,3)\}\subseteq [3]\times[3]$ and let $\sigma_1=(1,3)$, $\sigma_2=(2,3)$. Then, $\sigma_1(s)=\{(2,3), (3,2), (1,1)\}$, $\sigma_2(s)=\{(3,1), (1,3), (2,2)\}$.

Together with Theorem~\ref{fixedpoint}, the following theorem drastically reduces the dimension of the parameter space for the purpose of solving $\mathcal{P}_{\mathcal{S}_k}$.
\begin{thm}\label{2-dim}
$\mathcal{P}_{\mathcal{S}_k}$ has an optimal point within $\mathfrak S_k^{fix}$.
\end{thm}

\begin{proof}
Suppose that $\bar{\mathcal{C}} \in \mathfrak S_k$ achieves the optimal value of $\mathcal{P}_{\mathcal{S}_k}$. By Definition~\ref{fixed-point-definition}, for any $\sigma \in Sym(k)$, $\sigma(\bar{\mathcal{C}})$ is an optimal point of $\mathcal{P}_{\mathcal{S}_k}$. Let
$$
\hat{\mathcal{C}}=\frac{\sum_{\sigma\in Sym(k)}\sigma(\bar{\mathcal{C}})}{n!}
$$
It is easy to see that for any $\sigma \in Sym(k)$, $\sigma(\hat{\mathcal{C}})=\hat{\mathcal{C}}$. Hence, $\hat{\mathcal{C}} \in \mathfrak S_k^{fix}$. On the other hand, it follows from Lemma~\ref{convex} that
$$
g_{\mathcal{S}_k}(\hat{\mathcal{C}})\leq \frac{\sum_{\sigma \in Sym(k)} g_{\mathcal{S}_k}(\sigma(\bar{\mathcal{C}}))}{n!}
$$
and hence $\hat{\mathcal{C}}$ is an optimal point, which completes the proof.
\end{proof}

\subsection{Asymptotics of $\mathcal P_{\mathcal S_k}$}\label{subsection-limit}

For a $k$-sample $s=\{(i_1, j_1), (i_2, j_2), \dots, (i_{\alpha(s)}, j_{\alpha(s)})\}$,
we define the following multi-set:
$$
Ind_s:=\{i_1, j_1, i_2, j_2, \dots, i_{\alpha(s)}, j_{\alpha(s)}\},
$$
where $\alpha(s)$ denotes the size of $s$. And for any $\ell=1, 2, \dots, k$, denote by $m_{Ind_s}(\ell)$ the multiplicity of $\ell$ in $Ind_s$ (if $\ell\notin Ind_s$, then $m_{Ind_s}(\ell)=0$), and define
$$
\beta(s):=|\{\ell:  m_{Ind_s}(\ell)\neq 0\}|.
$$
For a quick example, consider $s=\{(1,1),(2,2),(1,3),(3,4),(1,6)\}\subseteq[6]\times[6]$. Then, $Ind_s=\{1,1,2,2,1,3,3,4,1,6\}$, $m_{Ind_s}(1)=4$, $m_{Ind_s}(2)=m_{Ind_s}(3)=2$, $m_{Ind_s}(4)=m_{Ind_s}(6)=1$, $m_{Ind_s}(5)=0$ and $\alpha(s)=\beta(s)=5$.

In this section, we characterize the asymptotics of $\{\mathcal{O}_{k}\}$ and thereby approximately ``solve'' $\mathcal{P}_{\mathcal{S}_k}$ for large $k$. We first recall the following theorem from~\cite{CH17}.
\begin{thm}\label{bounded-1}
$\mathcal{O}_{\mathcal{S}_k}\leq \frac{9}{8}$ for $k\geq 3$.
\end{thm}

By Theorem \ref{2-dim}, there exists an optimal point $\mathcal C_k\in \mathfrak S^{fix}_k$ for $\mathcal{P}_{\mathcal{S}_k}$. Moreover, by Theorem \ref{fixedpoint}, we can assume $\mathcal C_k=((c^{(1)}_{i,j}), (c^{(2)}_{i,j}), \dots, (c^{(k)}_{i,j}))$ takes the form as in (\ref{fixed-form}) with $a, b, x, y$ replaced by $a_k, b_k, x_k, y_k$, respectively, to emphasize its dependence on $k$, that is,
\begin{equation}\label{opt-point}
c^{(\ell)}_{i,j}=\left\{
              \begin{array}{lll}
               x_k, & \hbox{if $i=j=\ell$,}\\
               y_k, & \hbox{if $i=j\neq \ell$}\\
               a_k, & \hbox{if $i=\ell;j\neq \ell$ or $j=\ell; i\neq \ell$,}\\
               b_k, & \hbox{if otherwise,}
              \end{array} \right.
\end{equation}
where $x_k+(k-1)y_k=1$ and $y_k+(k-2)b_k+a_k=0$.

\begin{lem}\label{eval-entry}
For $x_k, y_k, a_k, b_k$ defined in (\ref{opt-point}), we have
\begin{itemize}
  \item [1)] $y_k=O(\frac{1}{k^2})$;
  \item [2)] $x_k=O(\frac{1}{k})$;
  \item [3)] $a_k=\frac{1}{k}+O(\frac{1}{k^2})$;
  \item [4)] $b_k=\frac{-1}{k^2}+O(\frac{1}{k^3})$.
\end{itemize}
\end{lem}

\begin{proof}
By definition, for $s=\{(1,1), (2,2), \dots, (\ell, \ell)\}\in \mathcal{S}_k$ and the optimal point $\mathcal C_k$ defined in (\ref{opt-point}), we have
\begin{equation}\label{value-diag-sample}
\begin{split}
g_s(\mathcal C_k)=\ell|x_k+(\ell-1)y_k|+(k-\ell)|\ell y_k|.
\end{split}
\end{equation}

Taking $\ell=k/2$ in (\ref{value-diag-sample}) and applying Theorem~\ref{bounded-1}, we have
$$
\frac{k^2}{4}|y_k|\leq g_s(\mathcal C_k)\leq \mathcal{O}_{\mathcal{S}_k}=O(1),
$$
which implies $y_k=O(\frac{1}{k^2})$. Hence $1)$ holds.

Taking $\ell=k$ in (\ref{value-diag-sample}) and applying Theorem~\ref{bounded-1}, we have
$$
k|x_k+(k-1)y_k|\leq \mathcal{O}_{\mathcal{S}_k}=O(1).
$$
Then, from $1)$ we deduce that $y_k=O(\frac{1}{k^2})$, which further implies $x_k=O(\frac{1}{k})$ by the above equation. Hence $2)$ holds.

Noticing that $(k-1)a_k=1-x_k$ and by $2)$, we have $a_k=\frac{1}{k}+O(\frac{1}{k^2})$. Hence $3)$ holds.

Noticing that $(k-2)b_k=-a_k-y_k$ and by $3)$ and $1)$, we have $b_k=\frac{-1}{k^2}+O(\frac{1}{k^3})$. Hence $4)$ holds.
\end{proof}

Now, we are ready to give the main result of this section.
\begin{thm}\label{lim-1}
$$
\lim_{k\rightarrow\infty}\mathcal{O}_{\mathcal{S}_k}=\frac{9}{8}.
$$
\end{thm}

\begin{proof}
Let $s=\{(i_1,1),(i_2,2),\dots,(i_\ell, \ell)\}\in \mathcal{S}_k$ be such that $\{i_1, i_2,\dots,i_\ell\}=\{1,2,\dots,\ell\}$ and $i_j\neq j$ for $j=1,2,\dots,\ell$. It can be easily verified that $Ind_s=\{1,1,2,2,\dots, \ell,\ell\}$. Let $\mathcal C_k$ be an optimal point of $\mathcal{P}_{\mathcal{S}_k}$ taking the form in (\ref{opt-point}). Then, by definition, we have
\begin{equation*}
\begin{split}
g_s(\mathcal C_k)=&\sum_{i=1}^k|m_{Ind_s}(i)a_k+(\ell-m_{Ind_s}(i))b_k|\\
=&\sum_{i\in Ind_s}|m_{Ind_s}(i)a_k+(\ell-m_{Ind_s}(i))b_k|+\sum_{i\notin Ind_s}|\ell b_k|\\
=&\sum_{i\in Ind_s}|2a_k+(\ell-2)b_k|+\sum_{i\notin Ind_s}|\ell b_k|\\
=&\ell|2a_k+(\ell-2)b_k|+(k-\ell)|\ell b_k|.
\end{split}
\end{equation*}
It then follows from Lemma \ref{eval-entry} that $2a_k+(\ell-2)b_k>0$ and $b_k<0$, and furthermore,
\begin{equation}\label{sample-9-8}
\begin{split}
g_s(\mathcal C_k)=&\ell|2a_k+(\ell-2)b_k|+(k-\ell)|\ell b_k|\\
=&\ell(2a_k+(2\ell-k-2)b_k)\\
=&\ell\left(\frac{2}{k}+O\left(\frac{1}{k^2}\right)-\frac{2\ell}{k^2}+\frac{1}{k}+O\left(\frac{1}{k^2}\right)\right)\\
=&\frac{\ell}{k^2}(3k-2\ell+O(1)).
\end{split}
\end{equation}
Now, setting $\ell=\frac{3k}{4}+O(1)$ in Equation (\ref{sample-9-8}), we have
$$
g_s(\mathcal C_k)=\frac{1}{k^2}\left(\frac{3k}{4}+O(1)\right)\left(\frac{3k}{2}+O(1)\right)=\frac{9}{8}+O\left(\frac{1}{k}\right).
$$
Hence, $\mathcal{O}_{\mathcal{S}_k}\geq g_s(\mathcal C_k)=\frac{9}{8}+O(\frac{1}{k})$. On the other hand, by Lemma \ref{bounded-1}, we have that $\mathcal{O}_{\mathcal{S}_k}\leq\frac{9}{8}$, which immediately implies that
$$
\lim_{k\rightarrow\infty}\mathcal{O}_{\mathcal{S}_k}=\frac{9}{8}.
$$
\end{proof}

\section{The Strong Homogeneous Flow $\mathcal C^*_k$ }\label{section-shf}

We introduce in this section a sequence of the so-called strong homogeneous flows $\{\mathcal C^*_k\}$. We will show that it is asymptotically optimal for $\{\mathcal P_{\mathcal S_k}\}$, yet it only yield the exact optimal solution if and if only $k=1, 2, 6, 10$.  We note that $\{\mathcal C^*_k\}$ will also play important roles in terms of obtaining the exact optimal solutions; more specifically, as will be shown in Section~\ref{section-perturb}, the optimal solution $\mathcal C^{**}_k$, $k=3, 4, 5, 7, 8, 9$, are obtained using a perturbation from the corresponding $\mathcal C^*_k$.

\subsection{Asymptotic Optimality}

\begin{de}[{\bf Strong Homogeneous Flow}]\label{ex-s-flow}
Let
\begin{equation} \label{C-star-k}
\mathcal C^*_k :=((c^{*(1)}_{i,j}), (c^{*(2)}_{i,j}), \dots, (c^{*(k)}_{i,j}))
\end{equation}
where
\begin{equation}
c^{*(\ell)}_{i,j}=\left\{
              \begin{array}{lll}
               \frac{2}{k}-\frac{1}{k^2}, & \hbox{if $i=j=\ell$,}\\
               \frac{1}{k}-\frac{1}{k^2}, & \hbox{if $i=\ell;j\neq \ell$ or $j=\ell; i\neq \ell$,}\\
               -\frac{1}{k^2}, & \hbox{$i\neq \ell$ and $j\neq \ell$.}
              \end{array} \right.
\end{equation}
In the remainder of this paper, $\mathcal C^*_k$ will be referred to as the {\em strong homogeneous flow}. Here we note that $\mathcal{C}^*$ can be alternatively obtained by combining and concatenating elementary flows as in (IV.1) of~\cite{CH15}.
\end{de}

\begin{exmp}
By definition, we have $\mathcal C^*_1=((1))$ and
$$
\mathcal C^*_2=\left(\left(
\begin{array}{cc}
\frac{3}{4} & \frac{1}{4} \\
\frac{1}{4} &\frac{-1}{4} \\
\end{array}
\right),
\left(\begin{array}{cc}
\frac{-1}{4} & \frac{1}{4} \\
\frac{1}{4} & \frac{3}{4} \\
\end{array}
\right)\right),
$$
$$
\mathcal C^*_3=\left( \left(
    \begin{array}{ccc}
      \frac{5}{9} & \frac{2}{9} & \frac{2}{9} \\
      \frac{2}{9} &\frac{-1}{9} &\frac{-1}{9} \\
      \frac{2}{9} &\frac{-1}{9} &\frac{-1}{9} \\
    \end{array}
  \right),\left(
    \begin{array}{ccc}
\frac{-1}{9} & \frac{2}{9} &\frac{-1}{9} \\
      \frac{2}{9} & \frac{5}{9} & \frac{2}{9} \\
\frac{-1}{9} & \frac{2}{9} &\frac{-1}{9} \\
    \end{array}
  \right),\left(
    \begin{array}{ccc}
\frac{-1}{9} &\frac{-1}{9} & \frac{2}{9} \\
\frac{-1}{9} &\frac{-1}{9} & \frac{2}{9} \\
      \frac{2}{9} & \frac{2}{9} & \frac{5}{9} \\
    \end{array}
  \right)\right).
$$
Note that $\mathcal C^*_2$ is the unique optimal point for $\mathcal P_{\mathcal S_2}$.
\end{exmp}

The following observation in \cite{CH15} will serve as a key lemma in this paper.

\begin{lem}[\cite{CH15}]\label{flow-value}
Let $\mathcal C^*_k$ be the strong homogeneous multi-flow and $\mathcal{S}_k$ be the strongly reachable sample set. Then, for all $s\in \mathcal{S}_k$,
\begin{equation}\label{Eq-flow-value}
\begin{split}
g_s(\mathcal C^*_k)=\frac{3k\alpha(s)-2\beta(s)\alpha(s)}{k^2}.
\end{split}
\end{equation}
\end{lem}

We now define
$$
\mathcal{S}_k(a, b):=\{s\in \mathcal{S}_k\|\alpha(s)=a, \beta(s)=b\}.
$$
The following two lemmas follow from Lemma~\ref{flow-value} via straightforward computations.

\begin{lem}\label{sub-opt-value}
For $\ell=1,2,\dots$, we have
\begin{description}
  \item[1)] If $k=4\ell$, then $g_{s}(\mathcal C^*_k)$ reaches the maximum $\frac{9}{8}$ when $s\in \mathcal{S}_k(3\ell, 3\ell)$;
  \item[2)] If $k=4\ell+1$, then $g_s(\mathcal C^*_k)$ reaches the maximum $\frac{18\ell^2+9\ell+1}{16\ell^2+8\ell+1}$ when $s\in \mathcal{S}_k(3\ell+1, 3\ell+1)$;
  \item[3)] If $k=4\ell+2$, then $g_s(\mathcal C^*_k)$ reaches the maximum $\frac{9\ell^2+9\ell+2}{8\ell^2+8\ell+2}$ when $s\in \mathcal{S}_k(3\ell+1, 3\ell+1)\cup \mathcal{S}_k(3\ell+2, 3\ell+2)$;
  \item[4)] If $k=4\ell+3$, then $g_s(\mathcal C^*_k)$ reaches the maximum $\frac{18\ell^2+27\ell+10}{16\ell^2+24\ell+9}$ when $s\in \mathcal{S}_k(3\ell+2, 3\ell+2)$.
\end{description}
\end{lem}

\begin{lem}\label{sec-max-value}
For $\ell=1,2\dots$, we have
\begin{description}
  \item[1)] If $k=4\ell$, then $g_{s}(\mathcal C^*_k)$ reaches the second largest value $\frac{9}{8}-\frac{2}{k^2}$ when $s\in \mathcal{S}_k(3\ell-1, 3\ell-1)\cup \mathcal{S}_k(3\ell+1, 3\ell+1)$;
  \item[2)] If $k=4\ell+1$, then $g_s(\mathcal C^*_k)$ reaches the second largest value  $\frac{18\ell^2+9\ell+1}{16\ell^2+8\ell+1}-\frac{1}{k^2}$ when $s\in \mathcal{S}_k(3\ell, 3\ell)$;
  \item[3)] If $k=4\ell+2$, then $g_s(\mathcal C^*_k)$ reaches the second largest value $\frac{9\ell^2+9\ell+2}{8\ell^2+8\ell+2}-\frac{4}{k^2}$ when $s\in \mathcal{S}_k(3\ell, 3\ell)\cup \mathcal{S}_k(3\ell+3, 3\ell+3)$;
  \item[4)] If $k=4\ell+3$, then $g_s(\mathcal C^*_k)$ reaches the second largest value $\frac{18\ell^2+27\ell+10}{16\ell^2+24\ell+9}-\frac{1}{k^2}$ when $s \in \mathcal{S}_k(3\ell+3, 3\ell+3)$.
\end{description}
\end{lem}

\begin{de}[{\bf Asymptotically Optimal Solution}]
A sequence $\{\mathcal C_k \| C_k \in \mathfrak S_k\}$ is said to be {\em asymptotically optimal} for $\{\mathcal{P}_{\mathcal{S}_k}\}$ if
$$
\lim_{k\rightarrow\infty}g_{\mathcal{S}_k}(\mathcal C_k)=\lim_{k\rightarrow\infty} \mathcal{O}_{\mathcal{S}_k}.
$$
\end{de}

The following theorem then immediately follows from Lemma~\ref{sub-opt-value}:
\begin{thm}
$\{\mathcal C^*_k\}$ is asymptotically optimal for $\{\mathcal{P}_{\mathcal{S}_k}\}$.
\end{thm}

\subsection{Optimality of $\mathcal C^*_6$ and $\mathcal C^*_{10}$}

In this section, we prove that $\mathcal C^*_k$ is an optimal solution to $\mathcal P_{\mathcal S_k}$ if and only if $k=1, 2, 6, 10$. We first state some needed notations and lemmas.

For any $\mathcal{C} \in \mathfrak S_k$, let $\mathcal{S}^{\dagger}_k(\mathcal{C})$ denote the set of all $k$-sample $s$ such that
\begin{equation}\label{canonical-signal-property}
\left\{
  \begin{array}{ll}
    g^{(\ell)}_s(\mathcal C)>0, & \hbox{if $\ell\in Ind_s$,} \\
    g^{(\ell)}_s(\mathcal C)<0, & \hbox{if $\ell\notin Ind_s$.}
  \end{array}
\right.
\end{equation}
We then have the following lemma.
\begin{lem}\label{signal-local}
For any $d$, we have
$$
\mathcal{S}_k(d, d) \subset \mathcal{S}^{\dagger}_k(\mathcal{C}_k^*).
$$
\end{lem}

\begin{proof}
Notice that for each $s\in \mathcal{S}_k(d, d)$, $\alpha(s)=d<k$ and hence $g^{(\ell)}_s(\mathcal C^*_k)$ is the sum of at most $k-1$ entries of $(c_{i, j}^{*(\ell)})$. By the definition of $\mathcal C^*_k$, we infer that if $\ell\in Ind_s$, then there exists at least one entry with value $(k-1)/k^2$ or $(2k-1)/k^2$  and the sum of the other entries are greater than or equal to $-(k-2)/k^2$ and hence $g^{(\ell)}_s(\mathcal C^*_k)>0$; and if $\ell \notin Ind_s$, then obviously $g^{(\ell)}_s(\mathcal C^*_k)=-d/k^2 < 0$.
\end{proof}

An element in a $k$-sample $s$ is said to be {\em diagonal} if its two coordinates are the same, otherwise {\em non-diagonal}. Let $\gamma(s)$ denote the number of diagonal elements in $s$. For example, let $s=\{(1,1),(3,3),(1,2)(1,4),(2,5)\}$ be a $5$-sample. Then, $(1,1),(3,3)$ are diagonal $5$-samples, whereas $(1,2),(1,4),(2,5)$ are non-diagonal $5$-samples, and furthermore $\gamma(s)=2$.

\begin{lem}\label{val-dif-sam}
For any $s \in \mathcal{S}_k(d, d)$ with $d<k$, there exists a neighborhood $N(\mathcal C^*_k, \varepsilon)\subset\mathfrak S_k^{fix}$ of $\mathcal C^*_k$ such that for all $\mathcal C\in N(\mathcal C^*_k, \varepsilon)$,
\begin{equation} \label{from-Cstar}
g_s(\mathcal C)-g_s(\mathcal C^*_k)=\gamma(s)(\overline{x}+(2d-k-1)\overline{y})+(d-\gamma(s))((2\overline{a}+(2d-k-2)\overline{b})),
\end{equation}
where $\overline{a},\overline{b},\overline{x},\overline{y}$ are defined by
$$
\mathcal C-\mathcal C^*_k=\left(\left(
              \begin{array}{ccccc}
                \overline{x} & \overline{a} & \overline{a} & \dots & \overline{a} \\
                \overline{a} & \overline{y} & \overline{b} & \dots & \overline{b} \\
                \overline{a} & \overline{b} & \overline{y} & \dots & \overline{b} \\
                \vdots & \vdots &  & \ddots & \vdots \\
                \overline{a} & \overline{b} & \overline{b} &\dots & \overline{y} \\
              \end{array}
            \right), \left(
              \begin{array}{ccccc}
                \overline{y} & \overline{a} & \overline{b} & \dots & \overline{b} \\
                \overline{a} & \overline{x} & \overline{a} & \dots & \overline{a} \\
                \overline{b} & \overline{a} & \overline{y} & \dots & \overline{b} \\
                \vdots & \vdots &  & \ddots & \vdots \\
                \overline{b} & \overline{a} & \overline{b} &\dots & \overline{y} \\
              \end{array}
            \right),\dots, \left(
              \begin{array}{ccccc}
                \overline{y} & \overline{b} & \overline{b} & \dots & \overline{a} \\
                \overline{b} & \overline{y} & \overline{b} & \dots & \overline{a} \\
                \overline{b} & \overline{b} & \overline{y} & \dots & \overline{a} \\
                \vdots & \vdots &  & \ddots & \vdots \\
                \overline{a} & \overline{a} & \overline{a} &\dots & \overline{x} \\
              \end{array}
            \right)\right).
$$
\end{lem}

\begin{proof}
Recall from Lemma \ref{signal-local} that $g^{(\ell)}_s(\mathcal C^*_k)>0$ if $\ell\in Ind_s$ and $g^{(\ell)}_s(\mathcal C^*_k)<0$ if $\ell\notin Ind_s$. Since each function $g^{(\ell)}_s(\cdot)$ is continuous, there exists a sufficiently small $\varepsilon$ such that for all $\mathcal C\in N(\mathcal C^*_k, \varepsilon)$ and all $s\in \mathcal{S}_k(d, d)$, $g^{(\ell)}_s(\mathcal C)>0$ if $\ell\in Ind_s$ and $g^{(\ell)}_s(\mathcal C)<0$ if $\ell\notin Ind_s$. For any $\mathcal C\in N(\mathcal C^*_k, \varepsilon)$, we have
\begin{equation}\label{Eq-dif}
\begin{split}
g_s(\mathcal C)-g_s(\mathcal C^*_k)=&\sum_{\ell=1}^kg^{(\ell)}_s(\mathcal C)-\sum_{\ell=1}^kg^{(\ell)}_s(\mathcal C^*_k)\\
                                   =&\sum_{\ell\in Ind_s}(g^{(\ell)}_s(\mathcal C)-g^{(\ell)}_s(\mathcal C^*_k))+\sum_{\ell\notin Ind_s}(g^{(\ell)}_s(\mathcal C^*_k)-g^{(\ell)}_s(\mathcal C)).
\end{split}
\end{equation}
Noticing that $\alpha(s)=\beta(s)=d$, it is easy to check that
\begin{equation} \label{Eq-in-index}
\sum_{\ell\in Ind_s}(g^{(\ell)}_s(\mathcal C)-g^{(\ell)}_s(\mathcal C^*_k))=\gamma(s)\overline{x}+(d-1)\gamma(s)\overline{y}
+2(d-\gamma(s))\overline{a}+(d-2)(d-\gamma(s))\overline{b}
\end{equation}
and
\begin{equation}\label{Eq-out-index}
\begin{split}
\sum_{\ell\notin Ind_s}(g^{(\ell)}_s(\mathcal C^*_k)-g^{(\ell)}_s(\mathcal C))=(d-k)(\gamma(s)\overline{y}+(d-\gamma(s))\overline{b}).
\end{split}
\end{equation}
Combining Equations (\ref{Eq-in-index}) and (\ref{Eq-out-index}) and plugging the results into (\ref{Eq-dif}), we have
$$g_s(\mathcal C)-g_s(\mathcal C^*_k)=\gamma(s)(\overline{x}+(2d-k-1)\overline{y})+(d-\gamma(s))(2\overline{a}+(2d-k-2)\overline{b}),
$$
which completes the proof.
\end{proof}

\begin{lem}\label{solu-all-s}
For any fixed $d<k$, there exists $\mathcal C\in \mathfrak S^{fix}_k$ such that $g_s(\mathcal C)-g_s(\mathcal C^*_k)<0$ for all $s\in \mathcal{S}_k(d, d)$.
\end{lem}

\begin{proof}
Let $N(\mathcal C^*_k, \varepsilon)$ be the neighborhood of $\mathcal C^*_k$ as in Lemma \ref{val-dif-sam}. We will prove that there exists $\mathcal C\in N(\mathcal C^*_k, \varepsilon)$ such that $g_s(\mathcal C)-g_s(\mathcal C^*_k)<0$ for all $s\in \mathcal{S}_k(d, d)$. Note that, by Lemma \ref{val-dif-sam}, we only need to prove that there exist sufficiently small $\overline{a},\overline{b},\overline{x},\overline{y}$ satisfying the following system:
$$
\left\{
  \begin{array}{ll}
    i (\overline{x}+(2d-k-1)\overline{y}) + (d-i) (2\overline{a}+(2d-k-2)\overline{b}), \quad i=0, 1, \dots, d,\\
    \overline{x}+(k-1)\overline{a}=0,\\
    \overline{y}+(k-2)\overline{b}+\overline{a}=0.
  \end{array}
\right.
$$
Since the first and the $(d+1)$-th inequalities imply the second to the $d$-th inequalities, we only need to prove there exist sufficiently small $\overline{a},\overline{b},\overline{x},\overline{y}$ satisfying the following system:
\begin{equation}\label{Eq-system}
\left\{
  \begin{array}{ll}
    2\overline{a}+(2d-k-2)\overline{b}<0, \\
    \overline{x}+(2d-k-1)\overline{y}<0,\\
    \overline{x}+(k-1)\overline{a}=0,\\
    \overline{y}+(k-2)\overline{b}+\overline{a}=0,
  \end{array}
\right.
\end{equation}
or equivalently,
\begin{equation}\label{Eq-system-simp}
\left\{
  \begin{array}{ll}
    2\overline{a}+(2d-k-2)\overline{b}<0, \\
    2\overline{a}+\frac{(2d-k-1)(k-2)}{d-1}\overline{b}>0.\\
    \end{array}
\right.
\end{equation}
Since $d<k$, we have $\frac{2d-k-2}{2d-k-1}\neq\frac{k-2}{d-1}$ and hence $2d-k-2\neq \frac{(2d-k-1)(k-2)}{d-1}$, which implies that there exist sufficiently small $\overline{a}$ and $\overline{b}$ such that (\ref{Eq-system-simp}) holds. By the last two equations of (\ref{Eq-system}), $\overline{x},\overline{y}$ can also be chosen sufficiently small, which implies there exists $\mathcal C\in N(\mathcal C^*_k, \varepsilon)$ such that $g_s(\mathcal C)-g_s(\mathcal C^*_k)<0$ for all $s\in \mathcal{S}_k(d, d)$, which completes the proof.
\end{proof}

In what follows, a $k$-sample $s\in \mathcal S_k$ is said to be {\em maximizing} at $\mathcal C$ if $g_s(\mathcal C)=g_{\mathcal S_k}(\mathcal C)$, and we will use $\mathcal{S}_k^{max}(\mathcal{C})$ denote the set of all maximizing $k$-samples at $\mathcal{C}$. For a quick example, by 1) of Lemma~\ref{sub-opt-value}, when $k=4\ell$, any $s \in \mathcal S_k(3\ell, 3\ell)$ is a maximizing sample at $\mathcal{C}_k^*$; and moreover, Lemma~\ref{sub-opt-value} implies that $\mathcal S_k(3\ell, 3\ell)$ is the set of all maximizing $k$-samples, i.e., $\mathcal{S}_k^{max}(\mathcal{C}_k^*)=\mathcal S_k(3\ell, 3\ell)$. Now, we are ready to give the main result of this section.

\begin{thm}\label{optimal-condition-theorem}
$\mathcal C^*_k$ is an optimal solution to $\mathcal{P}_{\mathcal{S}_k}$ if and only if $k=1,2,6,10$.
\end{thm}

\begin{proof}
The case $k=1$ is trivial and it is known~\cite{CH17} that $\mathcal C^*_2$ is an optimal point for $\mathcal{P}_2$. So we only need to prove $\mathcal C^*_6$ and $\mathcal C^*_{10}$ are respectively optimal points for $\mathcal{P}_6$ and $\mathcal{P}_{10}$, and $\mathcal C^*_k$ is not an optimal point for $\mathcal{P}_{\mathcal{S}_k}$ when $k \neq 1,2,6,10$.

In the remainder of the proof, we consider the following cases:

{\bf Case $1$:} $k > 1$ and $k \mod 4 \neq 2$. In this case, by Lemma~\ref{sub-opt-value}, there exists a $d$ such that $\mathcal{S}_k^{max}(\mathcal{C}_k^*)=\mathcal{S}_k(d, d)$. Then, by Lemma~\ref{solu-all-s}, for some sufficiently small $\varepsilon$, we can choose $\mathcal C \in N(\mathcal C^*_k, \varepsilon)$ such that for any maximizing $s$ at $\mathcal C^*_k$, the following two conditions hold: $(1)$ $g_s(\mathcal C)<g_s(\mathcal C^*_k)$; $(2)$ $s$ is also a maximizing sample at $\mathcal C$ (Here $(2)$ is true because $|\mathcal{S}_k|$ is finite, and the function $g_s(\cdot)$ is continuous over $\mathfrak S_k$ for each $s$). It then follows that $\max_{s\in \mathcal S_k}\{g_s(\mathcal{C})\}<\max_{s\in \mathcal S_k}\{g_s(\mathcal C^*_k)\}$, which means $\mathcal C^*_k$ is not an optimal point for $\mathcal{P}_{\mathcal{S}_k}$.

{\bf Case $2$:} $k=4\ell+2$ for some integer $\ell \geq 3$. In this case, $\mathcal C^*_k$ is not an optimal point for $\mathcal{P}_{\mathcal{S}_k}$. To prove this, we first show that for some sufficiently small $\varepsilon$, there exists $\mathcal{C} \in N(\mathcal C^*_k, \varepsilon)$ such that $g_s(\mathcal{C})-g_s(\mathcal{C}^*_k)<0$ for all $s \in \mathcal{S}_k(3\ell+1, 3\ell+1) \cup \mathcal{S}_k(3\ell+2, 3\ell+2)$ with $\ell \geq 3$. By Lemma~\ref{val-dif-sam}, we only need to prove that there exist sufficiently small $\overline{a},\overline{b},\overline{x},\overline{y}$ satisfying the following system:
\begin{equation}\label{Eq-2-system}
\hspace{-1cm} \left\{
  \begin{array}{ll}
  i (\overline{x}+(2(3\ell+1)-k-1)\overline{y}+(3\ell+1-i)(2\overline{a}+(2(3\ell+1)-k-2)\overline{b}))<0, \quad i=0, 1, \dots, 3 \ell +1,\\
  i (\overline{x}+(2(3\ell+2)-k-1)\overline{y}+(3\ell+2-i)(2\overline{a}+(2(3\ell+2)-k-2)\overline{b}))<0, \quad i=0, 1, \dots, 3 \ell +2,\\
  \overline{x}+(k-1)\overline{a}=0,\\
  \overline{y}+(k-2)\overline{b}+\overline{a}=0.
  \end{array}
\right.
\end{equation}
Applying a similar argument as in the proof of Lemma \ref{solu-all-s} and noticing that $k=4\ell+2$, we simplify the above system to
\begin{equation}\label{Eq-2-system-simp}
\left\{
  \begin{array}{ll}
    2\overline{a}+2(\ell-1) \overline{b}<0, \\
    2\overline{a}+2\ell\overline{b}<0, \\
    2\overline{a}+\frac{4(2\ell-1)}{3}\overline{b}>0,\\
    2\overline{a}+\frac{4\ell(2\ell+1)}{3\ell+1}\overline{b}>0.
    \end{array}
\right.
\end{equation}
One then verifies that for all $\delta>0$,
$$
\begin{cases}\label{Eq-2-system-simp-solution}
\overline{a}&=\frac{-11}{10}\ell\delta, \\
\overline{b}&=\delta,
\end{cases}
$$
is a solution to (\ref{Eq-2-system-simp}). Choosing $\delta > 0$ small enough, we deduce that for some sufficiently small $\varepsilon$ there exists $\mathcal{C} \in N(\mathcal C^*_k, \varepsilon)$ such that $g_s(\mathcal{C})-g_s(\mathcal C^*_k)<0$ for all $s\in \mathcal{S}_k(3\ell+1, 3\ell+1)\cup \mathcal{S}_k(3\ell+2, 3\ell+2)$ with $\ell \geq 3$. By the same reasoning as in Case $1$, for some sufficiently small $\varepsilon$, we can choose $\mathcal{C} \in N(\mathcal C^*_k, \varepsilon)$ such that for any $s\in \mathcal S_k^{max}(\mathcal C^*_k)$, the following two hold: $(1)$ $g_s(\mathcal{C})<g_s(\mathcal{C}^*_k)$; $(2)$ $s$ is a maximizing sample of $\mathcal{C}$. Hence $\max_{s\in \mathcal{S}_k}\{g_s(\mathcal{C})\}<\max_{s\in \mathcal{S}_k}\{g_s(\mathcal{C}^*_k)\}$, which means $\mathcal C^*_k$ is not an optimal point for $\mathcal{P}_{\mathcal{S}_k}$.

{\bf Case $3$:} $k=4\ell+2$ with $\ell=1$, i.e., $k=6$. In this case, consider Equation (\ref{Eq-2-system-simp}), which can be rewritten as
\begin{equation}\label{Eq-2-system-simp-1}
\left\{
  \begin{array}{ll}
    \overline{a}<0, \\
    \overline{a}+\overline{b}<0, \\
    3\overline{a}+2\overline{b}>0,\\
    2\overline{a}+3\overline{b}>0.\\
    \end{array}
\right.
\end{equation}
Note that this system has no solution because by the last two inequalities, we have $\overline{a}+\overline{b}>0$, which contradicts the second inequality. Hence, within $N(\mathcal C^*_k, \varepsilon)$ (defined in Lemma \ref{val-dif-sam}), there is no point $\mathcal C$ such that $g_s(\mathcal C)-g_s(\mathcal C^*_k)<0$ for all maximizing $s$ at $\mathcal C^*_6$ (Note that by Lemma~\ref{sub-opt-value}, the set of all such $s$ is $\mathcal{S}_6^{max}(\mathcal{C}^*_6)=\mathcal{S}_6(4, 4) \cup \mathcal{S}_6(5, 5)$). Hence, $\mathcal C^*_6$ is a local optimal point for $\mathcal{P}_{\mathcal{S}_6}$, and furthermore, by Lemma \ref{convex}, $\mathcal C^*_6$ is a global optimal point for $\mathcal{P}_{\mathcal{S}_6}$.

{\bf Case $4$:} $k=4\ell+2$ with $\ell=2$, i.e., $k=10$. In this case, consider Equation (\ref{Eq-2-system-simp}), which can be rewritten as
\begin{equation}\label{Eq-2-system-simp-1}
\left\{
  \begin{array}{ll}
    \overline{a}+\overline{b}<0, \\
    \overline{a}+2\overline{b}<0, \\
    \overline{a}+2\overline{b}>0,\\
    7\overline{a}+20\overline{b}>0.\\
    \end{array}
\right.
\end{equation}
Note that this system has no solution because the third inequality contradicts the second inequality. Hence, within $N(\mathcal C^*_k, \varepsilon)$ (defined in Lemma \ref{val-dif-sam}), there is no point $\mathcal C$ such that $g_s(\mathcal C)-g_s(\mathcal C^*_k)<0$ for all maximizing $s$ at $\mathcal{C}^*_{10}$ (Note that by Lemma~\ref{sub-opt-value}, the set of all such $s$ is $\mathcal{S}_{10}^{max}(\mathcal{C}^*_10)=\mathcal{S}_{10}(7, 7) \cup \mathcal{S}_{10}(8, 8)$). Hence, $\mathcal C^*_{10}$ is a local optimal point for $\mathcal{P}_{\mathcal{S}_{10}}$, and again by Lemma~\ref{convex}, a global optimal point for $\mathcal{P}_{\mathcal{S}_{10}}$,
\end{proof}

The following corollary says that the upper bound $\frac{9}{8}$ (derived in~\cite{CH17}) on $\mathcal{O}_{\mathcal S_k}$, $k \geq 3$, cannot be achieved.
\begin{co}
$\mathcal{O}_{\mathcal S_k}<\frac{9}{8}$ for all $k\geq 3$.
\end{co}

\section{A Perturbation Framework} \label{section-perturbation}

In this section, we propose a perturbation framework for the case $k \mod 4 \neq 2$ that not only promises a better solution to $\mathcal{P}_{\mathcal{S}_k}$ than $\mathcal{C}_k^*$ but also yields exact optimal solutions at least for some small $k$ (see Section~\ref{section-perturb} for exact solutions for the cases $k=3, 4, 5, 7, 8, 9$).

\subsection{Valid Perturbation Direction}

First of all, we define
$$
L(\mathfrak S^{fix}_k):=\{\mathbf{\Delta}=\mathcal C-\mathcal C'\|\mathcal C,\mathcal C'\in \mathfrak S^{fix}_k\}.
$$
Note that any $\mathbf{\Delta} = (\Delta^{(1)}, \Delta^{(2)}, \dots, \Delta^{(k)}) \in L(\mathfrak S^{fix}_k)$ can be written as
$$
\mathbf{\Delta}=\left(\left(
              \begin{array}{ccccc}
                \overline{x} & \overline{a} & \overline{a} & \dots & \overline{a} \\
                \overline{a} & \overline{y} & \overline{b} & \dots & \overline{b} \\
                \overline{a} & \overline{b} & \overline{y} & \dots & \overline{b} \\
                \vdots & \vdots &  & \ddots & \vdots \\
                \overline{a} & \overline{b} & \overline{b} &\dots & \overline{y} \\
              \end{array}
            \right), \left(
              \begin{array}{ccccc}
                \overline{y} & \overline{a} & \overline{b} & \dots & \overline{b} \\
                \overline{a} & \overline{x} & \overline{a} & \dots & \overline{a} \\
                \overline{b} & \overline{a} & \overline{y} & \dots & \overline{b} \\
                \vdots & \vdots &  & \ddots & \vdots \\
                \overline{b} & \overline{a} & \overline{b} &\dots & \overline{y} \\
              \end{array}
            \right),\dots, \left(
              \begin{array}{ccccc}
                \overline{y} & \overline{b} & \overline{b} & \dots & \overline{a} \\
                \overline{b} & \overline{y} & \overline{b} & \dots & \overline{a} \\
                \overline{b} & \overline{b} & \overline{y} & \dots & \overline{a} \\
                \vdots & \vdots &  & \ddots & \vdots \\
                \overline{a} & \overline{a} & \overline{a} &\dots & \overline{x} \\
              \end{array}
            \right)\right),
$$
where $\overline{x}(\mathbf{\Delta})+(k-1)\overline{a}(\mathbf{\Delta})=0$ and $\overline{y}(\mathbf{\Delta})+\overline{a}(\mathbf{\Delta})+(k-2)\overline{b}(\mathbf{\Delta})=0$. Here,
to emphasize the dependence, we have written $\overline{x},\overline{y},\overline{a},\overline{b}$ as $\overline{x}(\mathbf{\Delta})$, $\overline{y}(\mathbf{\Delta})$, $\overline{a}(\mathbf{\Delta})$, $\overline{b}(\mathbf{\Delta})$, respectively.

\begin{lem}\label{perturbation}
For any fixed $d<k$, there exist $\mathbf{\Delta}\in L(\mathfrak S^{fix}_k)$ and $\varepsilon_0 > 0$ such that for all $s_1,s_2\in \mathcal{S}_k(d, d)$ and all $0 < \varepsilon < \varepsilon_0$,
$$
g_{s_1}(\mathcal C^*_k+\varepsilon \mathbf{\Delta})-g_{s_1}(\mathcal C^*_k)=g_{s_2}(\mathcal C^*_k+\varepsilon \mathbf{\Delta})-g_{s_2}(\mathcal C^*_k)<0,
$$
where $\varepsilon \mathbf{\Delta} = (\varepsilon\mathbf{\Delta}^{(1)}, \varepsilon \mathbf{\Delta}^{(2)}, \dots, \varepsilon \mathbf{\Delta}^{(k)})$.
\end{lem}

\begin{proof}
To prove this lemma, we only need to slightly modify the proof of Lemma~\ref{solu-all-s}. More precisely, we assume $\overline{a}=(\frac{k^2}{2d}-k+1)\overline{b}$, which implies $2\overline{a}+(2d-k-2)\overline{b}=\overline{x}+(2d-k-1)\overline{y}$, which is an extra constraint added to (\ref{Eq-system}) ensuring (with the help of (\ref{from-Cstar})) a uniform change from $g_{s}(\mathcal C^*_k)$ to $g_{s}(\mathcal C^*_k+\varepsilon \mathbf{\Delta})$ over all $s \in \mathcal{S}_k(d, d)$, i.e., for all $s_1,s_2\in \mathcal{S}_k(d, d)$,
$$
g_{s_1}(\mathcal C^*_k+\varepsilon \mathbf{\Delta})-g_{s_1}(\mathcal C^*_k)=g_{s_2}(\mathcal C^*_k+\varepsilon \mathbf{\Delta})-g_{s_2}(\mathcal C^*_k),
$$
if $\varepsilon$ is small enough. Moreover, it can be readily verified that the new system is still solvable with the extra constraint. Finally, choosing a solution and then $\overline{a}$ (or equivalently, $\overline{b}$) properly similarly as in the proof of Lemma~\ref{solu-all-s} yields the desired $\mathbf{\Delta}$.
\end{proof}

\begin{de}[{\bf Valid Perturbation Direction}]\label{perturbation base}
For any $k \mod 4 \neq 2$, let $d$ be such that the set of all maximizing $k$-samples at $\mathcal{C}_k^*$ is $\mathcal{S}_k(d, d)$ (see Lemma~\ref{sub-opt-value}). There is a unique {\em valid perturbation direction} $\mathbf{\Delta}_k^* \in L(\mathfrak S^{fix}_k)$ such that 1) $|\overline{a}(\mathbf{\Delta}_k^*)|=1$; 2) $\overline{a}(\mathbf{\Delta}_k^*)=(\frac{k^2}{2d}-k+1)\overline{b}(\mathbf{\Delta}_k^*)$; 3) $2\overline{a}(\mathbf{\Delta}_k^*)+(2d-k-2)\overline{b}(\mathbf{\Delta}_k^*)<0$.
\end{de}

\begin{rem}
The ideas behind the above definition can be roughly explained as follows: As in the proof of Lemma~\ref{perturbation}, Conditions 2) and 3) will guarantee that the value of $g_s(\mathcal{C})$ uniformly decreases (over all $s \in \mathcal{S}_k(d, d)$) when perturbing $\mathcal{C}$ from $\mathcal{C}^*_k$ along the direction of $\mathbf{\Delta}_k^*$, and Condition 1) serves to ``normalize'' $\mathbf{\Delta}_k^*$ to yield the uniqueness.
\end{rem}

\begin{rem} \label{Delta-Star}
For $k \mod 4 = 2$, there are two distinct $d_1, d_2$ such that the set of all maximizing $k$-samples is $\mathcal{S}_k(d_1, d_1) \cup \mathcal{S}_k(d_2, d_2)$ (see Lemma~\ref{sub-opt-value}), and a perturbation direction that is valid with respect to $\mathcal{S}_k(d_1, d_1)$ may not be valid with respect to $\mathcal{S}_k(d_2, d_2)$. This is the key reason that our perturbation framework may not work for the case $k \mod 4 = 2$, since it requires a uniform (over all $s \in \mathcal{S}_k^{max}(\mathcal{C}_k^*)$) decrease of the maximum in the course of perturbation.
\end{rem}

\begin{exmp}\label{perturbation-base-k}
Let $k=3$ and $d=2$. It then follows from Condition $(2)$ of Definition \ref{perturbation base} that  $\overline{b}(\mathbf{\Delta}_3^*)=4\overline{a}(\mathbf{\Delta}_3^*)$. And by Condition $(3)$, we infer that $\overline{a}(\mathbf{\Delta}_3^*)>0$. Moreover, by Condition $(1)$, we have $\overline{a}(\mathbf{\Delta}_3^*)=1$, $\overline{b}(\mathbf{\Delta}_3^*)=4$, $\overline{x}(\mathbf{\Delta}_3^*)=-2$ and $\overline{y}(\mathbf{\Delta}_3^*)=-5$, or equivalently,
$$
\mathbf{\Delta}_3^*=\left(\left(
              \begin{array}{ccc}
                -2 & 1 & 1 \\
                1 & -5 & 4 \\
                1 & 4 & -5\\
              \end{array}
            \right), \left(
              \begin{array}{ccc}
                -5 & 1 & 4 \\
                1 & -2 & 1\\
                4 & 1 & -5 \\
              \end{array}
            \right),\left(
              \begin{array}{ccc}
                -5 & 4 & 1 \\
                4 & -5 & 1 \\
                1 & 1  & -2 \\
              \end{array}
            \right)\right).
$$
Similarly, set $k=4$ and $d=3$. Going through similar arguments as above, we have $\overline{a}(\mathbf{\Delta}_4^*)=-1$, $\overline{b}(\mathbf{\Delta}_4^*)=3$, $\overline{x}(\mathbf{\Delta}_4^*)=3$ and $\overline{y}(\mathbf{\Delta}_4^*)=-5$.
\end{exmp}

\subsection{Valid Perturbation Size}

In this section, assuming $k \mod 4 \neq 2$, we discuss the valid perturbation size for $\mathbf{\Delta}_k^*$. For notational convenience, we will henceforth write
$$
h_s^{(\ell)}(\varepsilon \mathbf{\Delta}_k^*):=g_{s}^{(\ell)}(\mathcal C^*_k+\varepsilon \mathbf{\Delta}_k^*)-g_s^{(\ell)}(\mathcal C^*_k), \quad h_s(\varepsilon \mathbf{\Delta}_k^*):=g_{s}(\mathcal C^*_k+ \varepsilon \mathbf{\Delta}_k^*)-g_s(\mathcal C^*_k).
$$

First of all, we need the following definition.
\begin{de}[{\bf Valid Perturbation Size}]\label{valid-epsilon}
For a given $k$-sample $s$, $\varepsilon>0$ is called $g^{(\ell)}_s$-{\em valid}, $\ell=1, 2, \dots, k$, if
$$
g^{(\ell)}_{s}(\mathcal C^*_k+ \varepsilon \mathbf{\Delta}_k^*)\cdot g^{(\ell)}_s(\mathcal C^*_k) \geq 0;
$$
and $\varepsilon$ is called $g_s$-{\em valid} if for all $1\leq \ell\leq k$, $\varepsilon$ is $g^{(\ell)}_s$-valid; $\varepsilon$ is called $g_{\mathcal{S}_k}$-{\em valid} if for all $s\in \mathcal{S}_k$, $\varepsilon$ is $g_s$-valid.
\end{de}

\begin{rem}
Since the function $g^{(\ell)}_s$ is continuous, there always exists $\varepsilon>0$ such that it is $g^{(\ell)}_s$-valid, and furthermore, there always exists $\varepsilon$ such that it is $g_s$-valid and $g_{\mathcal{S}_k}$-valid.
\end{rem}

We will also need the following two lemmas, whose proofs have been postponed to Appendices~\ref{proof-valid-k-3} and~\ref{proof-valid-k-4}, respectively.
\begin{lem}\label{valid-k-3}
$\varepsilon>0$ is $g_{\mathcal{S}_3}$-valid if and only $\varepsilon\leq\frac{1}{36}$.
\end{lem}

\begin{lem}\label{valid-k-4}
$\varepsilon>0$ is $g_{\mathcal{S}_4}$-valid if and only $\varepsilon\leq\frac{1}{176}$.
\end{lem}

\subsection{Formula of $h_{s}(\varepsilon \mathbf{\Delta}_k^*)$}

In this section, assuming $k \mod 4 \neq 2$, we will deduce a formula to compute $h_{s}(\varepsilon \mathbf{\Delta}_k^*)$ for $g_s$-valid perturbations.

We start with the following definition.
\begin{de}[{\bf Type of a Sample}] Let $s_{diag}$ and $s_{ndiag}$ denote the subsets of diagonal and non-diagonal elements of $s$, respectively. We say
\begin{equation}\label{type}
\left\{ {m_{Ind_{s_{diag}}}(1)\brack m_{Ind_{s_{ndiag}}}(1)}, {m_{Ind_{s_{diag}}}(2)\brack m_{Ind_{s_{ndiag}}}(2)},\dots, {m_{Ind_{s_{diag}}}(k)\brack m_{Ind_{s_{ndiag}}}(k)} \right\}
\end{equation}
is the \emph{type} of $s$, which will be denoted by $T(s)$. And slightly abusing the notation, we may also use (\ref{type}) to denote the set of all the samples of the type $T(s)$.
\end{de}

\begin{exmp}
For example, let $s=\{(1,1),(3,3),(1,2)(1,4),(2,5)\}$ be a $5$-sample. Then, $s_{diag}=\{(1,1),(3,3)\}$, $s_{ndiag}=\{(1,2),(1,4),(2,5)\}$, and $T(s)=\left\{{2\brack 2},{0\brack 2},{2\brack 0},{0\brack 1},{0\brack 1}\right\}$.
\end{exmp}

Note that if two $k$-samples $s_1,s_2$ are in the same orbit of $Sym(k)$, namely, there exists $\sigma\in Sym(k)$ such that $s_1=\sigma(s_2)$, then $T(s_1)=T(s_2)$, but the reverse direction does not hold in general. For example, one can check that $s_1=\{(1,1),(4,3),(3,4),(1,5)\}$ and $s_2=\{(1,1),(1,3),(3,4),(4,5)\}$ are not in the same orbit despite the fact they have the same type.

We have the following lemma, which says that for any $\mathcal{C} \in \mathfrak S^{fix}_k$, $g_s(\mathcal{C})$ is determined by $T(s)$. Note that the same statement may not hold true for $\mathcal C \notin \mathfrak S^{fix}_k$.
\begin{lem}
Let $\mathcal C\in \mathfrak S^{fix}_k$. Then, for any $k$-samples $s_1, s_2$ with $T(s_1)=T(s_2)$, we have $g_{s_1}(\mathcal C)=g_{s_2}(\mathcal C)$.
\end{lem}

\begin{proof}
Since $T(s_1)=T(s_2)$, we can find a $\sigma\in Sym(k)$ such that for all $i=1,2\dots,k$, $${m_{Ind_{{(s_1)}_{diag}}}(i)\brack m_{Ind_{{(s_1)}_{ndiag}}}(i)}={m_{Ind_{{(s_2)}_{diag}}}(\sigma(i))\brack m_{Ind_{{(s_2)}_{ndiag}}}(\sigma(i))}.$$
Since $\mathcal C\in \mathfrak S^{fix}_k$, we have $g^{(i)}_{s_1}(\mathcal C)=g^{(\sigma(i))}_{s_2}(\mathcal C)$. Hence,
$$
g_{s_1}(\mathcal C)=\sum_{i=1}^kg^{(i)}_{s_1}(\mathcal C)=\sum_{i=1}^kg^{(\sigma(i))}_{s_2}(\mathcal C)=g_{s_2}(\mathcal C),
$$
as desired.
\end{proof}

We also need the following definition, which can be used to give an alternative classification of samples.
\begin{de}[{\bf Discriminant}]
For any $k$ with $k \mod 4 \neq 2$, the {\em discriminant} of a $k$-sample $s$ is defined by
$$
\mathcal{D}_k(s):=\overline{a}(\mathbf{\Delta}_k^*)+\gamma(s)\overline{y}(\mathbf{\Delta}_k^*)
+(\alpha(s)-\gamma(s)-1)\overline{b}(\mathbf{\Delta}_k^*).
$$
\end{de}
\noindent We next give an example for the above definition, for which we need to introduce more notation as follows: Let
$$
\mathcal{S}_k(a,b,c):=\{s\in \mathcal{S}_k\|\alpha(s)=a, \beta(s)=b, \gamma(s)=c\},
$$
and
$$
\mathcal{S}_k(a,b,c,d):=\{s\in \mathcal{S}_k\|\alpha(s)=a, \beta(s)=b, \gamma(s)=c, \delta(s)=d\},
$$
where
$$
\delta(s):=|\{i\|m_{Ind_s}(i)=1\}|.
$$
For example, one verifies that for
$$
s=\{(1,1),(1,2),(2,3),(3,4),(6,5)\},
$$
we have $Ind_s=\{1,1,1,2,2,3,3,4,5,6\}$, and moreover, $m_{Ind_s}(4)=m_{Ind_s}(5)=m_{Ind_s}(6)=1$ and $\delta(s)=3$, which imply that $s \in \mathcal {S}_6(5,6,1,3)$.

\begin{exmp}\label{disctiminat-k}
For the first case in Example \ref{perturbation-base-k}, $\overline{x}(\mathbf{\Delta}_3^*)=-2$, $\overline{y}(\mathbf{\Delta}_3^*)=-5$, $\overline{a}(\mathbf{\Delta}_3^*)=1$ and $\overline{b}(\mathbf{\Delta}_3^*)=4$. Then, for any $s\in \mathcal{S}_3$,
$$
\mathcal{D}_3(s)=1-5\gamma(s)+4(\alpha(s)-\gamma(s)-1).
$$
More specifically,
\begin{itemize}
  \item if $s\in \mathcal{S}_3(3,3,0)$, then $\mathcal{D}_3(s)=9$;
  \item if $s\in \mathcal{S}_3(3,3,1)$, then $\mathcal{D}_3(s)=0$;
  \item if $s\in \mathcal{S}_3(3,3,2)$, then $\mathcal{D}_3(s)=-9$.
\end{itemize}

Similarly, for the second case in Example \ref{perturbation-base-k}, $\overline{x}(\mathbf{\Delta}_4^*)=3$, $\overline{y}(\mathbf{\Delta}_4^*)=-5$, $\overline{a}(\mathbf{\Delta}_4^*)=-1$ and $\overline{b}(\mathbf{\Delta}_4^*)=3$. Hence, for any $s \in \mathcal{S}_4$,
$$
\mathcal{D}_4(s)=-1-5\gamma(s)+3(\alpha(s)-\gamma(s)-1).
$$
More specifically,
\begin{itemize}
  \item if $s\in \mathcal{S}_4(4,4,0)$, then $\mathcal{D}_4(s)=8$;
  \item if $s\in \mathcal{S}_4(4,4,1)$, then $\mathcal{D}_4(s)=0$;
  \item if $s\in \mathcal{S}_4(4,4,2)$, then $\mathcal{D}_4(s)=-8$;
  \item if $s\in \mathcal{S}_4(4,4,3)$, then $\mathcal{D}_4(s)=-16$.
\end{itemize}
\end{exmp}

As will be shown below, the notion of discriminant can be used to give an alternative classification of samples.
\begin{de}[{\bf Class}]\label{defn-class}
A sample $s\in \mathcal{S}_k$ is said to be in class II if $s\in \mathcal{S}_k(k, k)$, $\delta(s)\neq 0$ and $\mathcal{D}_k(s)<0$. Otherwise, it is said to be in class I.
\end{de}

\begin{exmp}\label{class-k}
Using the fact $\mathcal{S}_3(3,3,2)=\mathcal{S}_3(3,3,2,1)$ and recalling Example \ref{disctiminat-k}, we have that $s\in \mathcal{S}_3$ is in class II if and only $s\in \mathcal{S}_3(3,3,2)$. Similarly, we have that $s\in \mathcal{S}_4$ is in class II if and only if $s\in \mathcal{S}_4(4,4,2) \cup \mathcal{S}_4(4,4,3)$ and $\delta(s)\neq 0$. It is easy to verify that
\begin{equation*}
\begin{split}
\mathcal{S}_4(4,4,2)&=\left\{{2 \brack 2},{2 \brack 0},{0 \brack 1},{0 \brack 1}\right\} \bigcup \left\{{2 \brack 1},{2 \brack 1},{0 \brack 1},{0 \brack 1}\right\}\\
&\bigcup\left\{{2 \brack 1},{2 \brack 0},{0 \brack 2},{0 \brack 1}\right\} \bigcup \left\{{2 \brack 0},{2 \brack 0},{0 \brack 2},{0 \brack 2} \right\},
\end{split}
\end{equation*}
where $\delta(s)=0$ if and only if $s\in \left\{{2 \brack 0},{2 \brack 0},{0 \brack 2},{0 \brack 2} \right\}$. This, together with the fact that $\mathcal{S}_4(4,4,3)=\mathcal{S}_4(4,4,3,1)$, implies that all the class II samples of $\mathcal{S}_4$ are
$$
\left\{{2 \brack 2},{2 \brack 0},{0 \brack 1},{0 \brack 1}\right\}\bigcup\left\{{2 \brack 1},{2 \brack 1},{0 \brack 1},{0 \brack 1}\right\}\bigcup\left\{{2 \brack 1},{2 \brack 0},{0 \brack 2},{0 \brack 1}\right\}\bigcup \mathcal{S}_4(4,4,3).
$$
\end{exmp}

The following lemma, which is the main result of this section, measures how much $g_s(\cdot)$ changes from $\mathcal{C}_k^*$ under a valid perturbation along the direction of $\mathbf{\Delta}_k^*$.
\begin{lem}\label{val-dif-sam-en}
Let $s\in \mathcal{S}_k$ and $\varepsilon>0$ be $g_s$-valid. Then,
\begin{equation}\label{Eq-val-dif-sam-en}
h_s(\varepsilon \mathbf{\Delta}_k^*)=A_s\cdot\overline{x}(\varepsilon \mathbf{\Delta}_k^*)
+B_s\cdot\overline{y}(\varepsilon \mathbf{\Delta}_k^*)+C_s\cdot\overline{a}(\varepsilon \mathbf{\Delta}_k^*)
+D_s\cdot\overline{b}(\varepsilon \mathbf{\Delta}_k^*),
\end{equation}
where, if $s$ is in class I, then
\begin{equation}\label{class-1}
\begin{split}
A_s&=\gamma(s),\\
B_s&=\gamma(s)(2\beta(s)-k-1),\\
C_s&=2(\alpha(s)-\gamma(s)),\\
D_s&=(\alpha(s)-\gamma(s))(2\beta(s)-k-2);
\end{split}
\end{equation}
and if $s$ is in class II, then
\begin{equation}\label{class-2}
\begin{split}
  A_s&=\gamma(s),\\
  B_s&=\gamma(s)(k-2\delta(s)-1),\\
  C_s&=2(k-\gamma(s)-\delta(s)),\\
  D_s&=k^2-(2\delta(s)+\gamma(s)+2)k+2\gamma(s)\delta(s)+2\gamma(s)+2\delta(s).
\end{split}
\end{equation}
\end{lem}

\begin{proof}
For any $s\in \mathcal{S}_k$ and any $1\leq \ell\leq k$, by the definition of $\mathcal C^*_k$, it is easy to see that \begin{equation}\label{g^l_s-all}
g_s^{(\ell)}(\mathcal C^*_k)=\frac{m_{Ind_s}(\ell)}{k}-\frac{\alpha(s)}{k^2},
\end{equation}
which immediately implies that
\begin{itemize}
  \item $g_s^{(\ell)}(\mathcal C^*_k)=0$ if $\alpha(s)=k$ and $m_{Ind_s}(\ell)=1$;
  \item $g_s^{(\ell)}(\mathcal C^*_k)<0$ if $m_{Ind_s}(\ell)=0$, i.e., $\ell\notin Ind_s$;
  \item $g_s^{(\ell)}(\mathcal C^*_k)>0$ otherwise.
\end{itemize}

We first consider the samples in class I. By definition, there are the following three cases: $(1)$ $s\notin \mathcal{S}_k(k, k)$; $(2)$ $s\in \mathcal{S}_k(k, k)$ and $\delta(s)=0$; $(3)$ $s\in \mathcal{S}_k(k, k)$, $\delta(s)\neq 0$ and $\mathcal{D}_k(s)\geq 0$.

By the above discussions, for Cases (1) and (2), we have $g_s^{(\ell)}(\mathcal C^*_k)\neq 0$ for all $1\leq \ell\leq k$. Then, by the definition of a valid perturbation, the following hold for the $g_s$-valid $\varepsilon$:
\begin{equation}\label{signal-valid-perturb}
\begin{split}
\text{if}\; g^{(\ell)}_{s}(\mathcal C^*_k)>0,\; &\text{then}\; g^{(\ell)}_{s}(\mathcal C^*_k+\varepsilon \mathbf{\Delta}_k^*)\geq0;\\
\text{if}\; g^{(\ell)}_{s}(\mathcal C^*_k)<0,\; &\text{then}\; g^{(\ell)}_{s}(\mathcal C^*_k+\varepsilon \mathbf{\Delta}_k^*)\leq0.
\end{split}
\end{equation}

For Case (3), since $\varepsilon$ is $g_s$-valid, (\ref{signal-valid-perturb}) still holds. In this case, since $\delta(s)\neq0$, there exists some $\ell$ such that $g^{(\ell)}_{s}(\mathcal C^*_k)=0$; and for such an $\ell$, it can be verified that 
\begin{equation} \label{signal-valid-perturb-1}
g^{(\ell)}_{s}(\mathcal C^*_k+\varepsilon \mathbf{\Delta}_k^*)
=\overline{a}(\varepsilon \mathbf{\Delta}_k^*)+\gamma(s)\overline{y}(\varepsilon \mathbf{\Delta}_k^*)
+(\alpha(s)-\gamma(s)-1)\overline{b}(\varepsilon \mathbf{\Delta}_k^*)=\mathcal{D}_k(s)\geq0.
\end{equation}
Now, combining (\ref{signal-valid-perturb}) and (\ref{signal-valid-perturb-1}), we deduce that
\begin{equation}\label{signal-valid-perturb-class-1}
\begin{split}
\text{if}\; \ell\in Ind_s,\; &\text{then}\; g^{(\ell)}_{s}(\mathcal C^*_k)\geq0\; \text{and}\; g^{(\ell)}_{s}(\mathcal C^*_k+\varepsilon \mathbf{\Delta}_k^*)\geq0;\\
\text{if}\; \ell\notin Ind_s,\; &\text{then}\; g^{(\ell)}_{s}(\mathcal C^*_k)<0\; \text{and}\; g^{(\ell)}_{s}(\mathcal C^*_k+\varepsilon \mathbf{\Delta}_k^*)\leq0.
\end{split}
\end{equation}
Hence, for any sample $s$ in class I, we have
\begin{equation}\label{Eq-class-1}
\begin{split}
h_s(\varepsilon \mathbf{\Delta}_k^*)&=\sum_{\ell=1}^k|g^{(\ell)}_{s}(\mathcal C^*_k+\varepsilon \mathbf{\Delta}_k^*)|-\sum_{\ell=1}^k|g^{(\ell)}_{s}(\mathcal C^*_k)|\\
&=\sum_{\ell\in Ind_s}(g^{(\ell)}_{s}(\mathcal C^*_k+\varepsilon \mathbf{\Delta}_k^*)-g^{(\ell)}_{s}(\mathcal C^*_k))-\sum_{\ell\notin Ind_s}(g^{(\ell)}_{s}(\mathcal C^*_k+\varepsilon \mathbf{\Delta}_k^*)-g^{(\ell)}_{s}(\mathcal C^*_k))\\
&=\sum_{\ell\in Ind_s} h^{(\ell)}_{s}(\varepsilon \mathbf{\Delta}^*_k)-\sum_{\ell\notin Ind_s} h^{(\ell)}_{s}(\varepsilon \mathbf{\Delta}^*_k).
\end{split}
\end{equation}
Note that
\begin{equation*}
\begin{split}
\sum_{\ell\in Ind_s} h^{(\ell)}_s(\varepsilon \mathbf{\Delta}_k^*) &=\gamma(s)\overline{x}(\varepsilon \mathbf{\Delta}_k^*)+
(\beta(s)-1)\gamma(s)\overline{y}(\varepsilon \mathbf{\Delta}_k^*)\\
&+2(\alpha(s)-\gamma(s))\overline{a}(\varepsilon \mathbf{\Delta}_k^*)+
(\alpha(s)-\gamma(s))(\beta(s)-2)\gamma(s)\overline{b}(\varepsilon \mathbf{\Delta}_k^*)
\end{split}
\end{equation*}
and
\begin{equation*}
\begin{split}
\sum_{\ell\notin Ind_s} h^{(\ell)}_s(\varepsilon \mathbf{\Delta}_k^*)&=(k-\beta(s))(\gamma(s)\overline{y}(\varepsilon \mathbf{\Delta}_k^*)+
(\alpha(s)-\gamma(s))\overline{b}(\varepsilon \mathbf{\Delta}_k^*)).
\end{split}
\end{equation*}
Substituting the above equalities into (\ref{Eq-class-1}) then yields the result for class I.

Now, we consider the samples in class II. By definition, there exists some $\ell$ such that $g^{(\ell)}_{s}(\mathcal C^*_k)=0$; and for such an $\ell$,
$$
g^{(\ell)}_{s}(\mathcal C^*_k+\varepsilon \mathbf{\Delta}_k^*)
=\overline{a}(\mathbf{\Delta}_k^*)+\gamma(s)\overline{y}(\mathbf{\Delta}_k^*)
+(\alpha(s)-\gamma(s)-1)\overline{b}(\mathbf{\Delta}_k^*)=\mathcal{D}_k(s)<0.
$$
Hence, similarly as above, we have
\begin{equation}\label{Eq-class-2}
\begin{split}
h_s(\varepsilon \mathbf{\Delta}^*_k)&=\sum_{\ell=1}^k|g^{(\ell)}_{s}(\mathcal C^*_k+\varepsilon \mathbf{\Delta}_k^*)|-\sum_{\ell=1}^k|h^{(\ell)}_{s}(\mathcal C^*_k)|\\
&=\sum_{\ell\in Ind_s} h^{(\ell)}_{s}(\varepsilon \mathbf{\Delta}^*_k)-\sum_{\ell\notin Ind_s} h^{(\ell)}_{s}(\varepsilon \mathbf{\Delta}^*_k)
-2\delta(s) \mathcal{D}_k(s).
\end{split}
\end{equation}
Noting that $\alpha(s)=\beta(s)=k$ for any class II sample $s$ and substituting for the values of $\sum_{\ell\in Ind_s}h^{(\ell)}_{s}(\varepsilon \mathbf{\Delta}^*_k)$, $\sum_{\ell\notin Ind_s} h^{(\ell)}_{s}(\varepsilon \mathbf{\Delta}^*_k)$ as in the proof for class I, the result for class II then follows, which completes the proof.
\end{proof}

\subsection{Perturbation framework} \label{framework}

Note that by Theorem~\ref{optimal-condition-theorem}, for any $k \neq 1, 2, 6, 10$, one can perturb $\mathcal{C}_k^*$ to obtain a better solution to $\mathcal{P}_{\mathcal{S}_k}$, which however may not be optimal. In the following, we propose a framework of perturbing $\mathcal{C}_k^*$ to obtain $\mathcal{C}_k^{**}$ for $k \mod \neq 2$, which are optimal at least for the cases $k = 3, 4, 5, 7, 9$ (see Section~\ref{section-perturb}).

\textbf{Step 1: Compute $\mathbf{\Delta}^*_k$.} This step can be done by solving 1), 2) and 3) in Definition~\ref{perturbation base}.

\textbf{Step 2: Compute $\mathcal C^{**}_k$.} For this step, we first use Lemmas~\ref{sub-opt-value} and~\ref{sec-max-value} to obtain the subsets of samples which achieves the maximum and the second largest values of $\{g_s(\mathcal C^*_k)\| s\in \mathcal S_k\}$. And we then use Lemma~\ref{val-dif-sam-en} to compute $h_s(\varepsilon\mathbf{\Delta}^*_k)$ for all $s \in \mathcal{S}_k$. In the end, we increase the value of $\varepsilon$ from $0$ so that the maximum will decrease (uniformly over all $s \in \mathcal{S}_k^{max}(\mathcal{C}_k^*)$) until it meets the increasing second largest value at $\varepsilon=\varepsilon^*$, and then set $\mathcal C^{**}_k=\mathcal C^*_k+\varepsilon^* \mathbf{\Delta}_k^*$.

\textbf{Step 3: Compute $\mathcal S_k^{max}(\mathcal C^{**}_k)$.} We first check by Definition \ref{valid-epsilon} the validity of $\varepsilon$ obtained in \textbf{Step 2}. It turns out that for each $k$, there might exist a small number of samples $s$ for which $\varepsilon$ is not $g_s$-valid. For such $s$, we can simply compute the value of $g_s(\mathcal C^{**}_k)$ using the definition of $g_s$, and then we compute, by using Lemma~\ref{flow-value} and Lemma~\ref{val-dif-sam-en}, the value of $g_s(\mathcal C^{**}_k)=g_s(\mathcal C^{*}_k)+h_s(\varepsilon\mathbf{\Delta}^*_k)$ for all $s$ where $\varepsilon$ is $g_s$-valid. Finally, with the values of all $g_s(\mathcal{C}_k^{**})$, we derive $\mathcal S_k^{max}(\mathcal C^{**}_k)$.

\section{Optimal Solutions for $k=3, 4, 5, 7, 8, 9$} \label{section-perturb}

In this section, through perturbing the corresponding $\mathcal C^*_k$, we obtain the optimal solutions $\mathcal C^{**}_k$ to $\mathcal{P}_{\mathcal{S}_k}$ for $k=3, 4, 5, 7, 8, 9$, and we further establish the uniqueness of these optimal solutions.

\subsection{From $\mathcal C^*_k$ to $\mathcal C^{**}_k$ for $k=3, 4, 5, 7, 9$}

The perturbation from $\mathcal C^*_k$ to $\mathcal C^{**}_k$ follows from the framework in Section~\ref{framework} with however some possible simplifications and adaptations to varying degrees for different $k$.

$\blacksquare$ We first deal with the case $k=3$ through the following steps.

\textbf{Step 1: Compute $\mathbf{\Delta}^*_3$.} This has already been done in Example~\ref{perturbation-base-k}.

\textbf{Step 2: Compute $\mathcal{C}_3^{**}$.} For this step, we need to compute $h_s(\varepsilon \mathbf{\Delta}^*_3)$ for all $s\in \mathcal{S}_3$. To this end, we compute using  Lemma~\ref{val-dif-sam-en},
$$
h_s(\varepsilon \mathbf{\Delta}^*_3)=(-2A_s-5B_s+C_s+4D_s)\varepsilon.
$$
By Example \ref{class-k}, $s$ is in class II if and only if $s\in \mathcal{S}_3(3,3,2,1)=\mathcal{S}_3(3,3,2)$. Then, by (\ref{class-2}), we have
$A_s=2$, $B_s=C_s=0$, $D_s=1$ and hence $h_s(\varepsilon \mathbf{\Delta}^*_3)=0$ for $s\in \mathcal{S}_3(3,3,2)$. For any sample $s$ in class I, we use (\ref{class-1}) to compute the coefficients and then compute $h_s(\varepsilon \mathbf{\Delta}^*_3)$. The computations as above yield Table \ref{table-k-3}, where the values of all $g_s(\mathcal C^*_3)$ and $h_s(\varepsilon \mathbf{\Delta}^*_3)$ are listed.

\begin{table}[h]
\centering
\caption{The values of $g_s(\mathcal C^*_3)$ and $h_s(\varepsilon \mathbf{\Delta}^*_3)$}
\vspace{0.5cm}
\begin{tabular}{|l||c|c|c|c|}
  \hline
\multicolumn{1}{|l||}{Class of $s$} &  \multicolumn{4}{|c|}{class I} \\ \hline
  Subclass of $s$ & $\mathcal{S}_3(1,1,1)$ & $\mathcal{S}_3(1,2,0)$ & $\mathcal{S}_3(2,2)$ & $\mathcal{S}_3(2,3,0)$ \\ \hline
$g_s(\mathcal C^*_3)$ & $\frac{7}{9}$ & $\frac{5}{9}$ & $\frac{10}{9}$ & $\frac{6}{9}$\\ \hline
  $h_s(\varepsilon \mathbf{\Delta}^*_3)$ & $8\varepsilon$ & $-2\varepsilon$ & $-4\varepsilon$ & $12\varepsilon$ \\
  \hline
\end{tabular}

\bigskip

\begin{tabular}{|l||c|c|c|c|c|}
  \hline
\multicolumn{1}{|l||}{Class of $s$} &  \multicolumn{4}{|c|}{class I} & \multicolumn{1}{|c|}{class II} \\ \hline
  Subclass of $s$ & $\mathcal{S}_3(2,3,1)$ & $\mathcal{S}_3(3,3,0)$ & $\mathcal{S}_3(3,3,1)$ & $\mathcal{S}_3(3,3,3)$ & $\mathcal{S}_3(3,3,2)$ \\ \hline
  $g_s(\mathcal C^*_3)$ & $\frac{6}{9}$ &  \multicolumn{4}{|c|}{1} \\ \hline
  $h_s(\varepsilon \mathbf{\Delta}^*_3)$ & $-6\varepsilon$ & $18\varepsilon$ & 0 & $-36\varepsilon$ & 0 \\  \hline
\end{tabular}
\label{table-k-3}
\end{table}

By Table~\ref{table-k-3}, $g_s(\mathcal C^*_3)$ achieves the maximum $\frac{10}{9}$ at $\mathcal{S}_3(2,2)$ (or, more precisely, at any sample from $\mathcal{S}_3(2,2)$) and the second largest value $1$ at $\mathcal{S}_3(3,3)$. Now, we will perturb $\mathcal{C}_k^*$ along the direction of $\mathbf{\Delta}_k^*$ to obtain $\mathcal{C}_k^{**}$ so that, roughly speaking, the maximum will decrease until it meets the increasing second largest value. To this end, we note that in the course of perturbation, the second largest value is always achieved at $\mathcal{S}_3(3,3)$, and we thereby solve $1+18\varepsilon=\frac{10}{9}-4\varepsilon$, which yields $\varepsilon^*=\frac{1}{22\times9}=\frac{1}{198}$ and furthermore,
$$
\mathcal C^{**}_3:=\mathbf{\Delta}_3^* \times \frac{1}{198}+\mathcal C^*_3=\left( \left(
    \begin{array}{ccc}
      \frac{12}{22} & \frac{5}{22} & \frac{5}{22} \\
      \frac{5}{22} &\frac{-3}{22} &\frac{-2}{22} \\
      \frac{5}{22} &\frac{-2}{22} &\frac{-3}{22} \\
    \end{array}
  \right),\left(
    \begin{array}{ccc}
      \frac{-3}{22} & \frac{5}{22} & \frac{-2}{22} \\
      \frac{5}{22} &\frac{12}{22} &\frac{5}{22} \\
      \frac{-2}{22} &\frac{5}{22} &\frac{-3}{22} \\
    \end{array}
  \right),\left(
    \begin{array}{ccc}
      \frac{-3}{22} & \frac{-2}{22} & \frac{5}{22} \\
      \frac{-2}{22} &\frac{-3}{22} &\frac{5}{22} \\
      \frac{5}{22} &\frac{5}{22} &\frac{12}{22} \\
    \end{array}
  \right)\right).
$$
By Lemma~\ref{valid-k-3}, $\mathbf{\Delta}_3^* \times \frac{1}{198}$ is a valid perturbation.

\textbf{Step 3: Compute $\mathcal S_3^{max}(\mathcal C^{**}_3)$.} From Table \ref{table-k-3}, it is easy to verify that $\{g_s(\mathcal C^{**}_3)\|s\in \mathcal{S}_3\}$ achieves the maximum $\frac{12}{11}$ at $\mathcal{S}_3(2,2)\cup \mathcal{S}_3(3,3,0)$. In other words, $\mathcal{S}_k^{max}(\mathcal{C}_3^{**})=\mathcal{S}_3(2,2)\cup \mathcal{S}_3(3,3,0)$.

$\blacksquare$ Now, we focus on the case $k=4$ through the following steps.

\textbf{Step 1: Compute $\mathbf{\Delta}^*_4$.} This has already been done in Example~\ref{perturbation-base-k}.

\textbf{Step 2: Compute $\mathcal{C}_4^{**}$.} For this step, we need to compute $h_s(\varepsilon \mathbf{\Delta}^*_4)$ for all $s\in \mathcal{S}_4$. To this end, we compute using Lemma~\ref{val-dif-sam-en},
$$
h_s(\varepsilon \mathbf{\Delta}^*_4)=(3A_s-5B_s-C_s+3D_s)\varepsilon.
$$
By Example \ref{class-k}, $s\in \mathcal{S}_4$ is in class II if and only if
$$
s\in\left\{{2 \brack 2},{2 \brack 0},{0 \brack 1},{0 \brack 1}\right\}\cup\left\{{2 \brack 1},{2 \brack 1},{0 \brack 1},{0 \brack 1}\right\}\cup\left\{{2 \brack 1},{2 \brack 0},{0 \brack 2},{0 \brack 1}\right\}\cup \mathcal{S}_4(4,4,3).
$$
For the class II samples, if $s\in\left\{{2 \brack 2},{2 \brack 0},{0 \brack 1},{0 \brack 1}\right\}\cup\left\{{2 \brack 1},{2 \brack 1},{0 \brack 1},{0 \brack 1}\right\}$, we have $\gamma(s)=\delta(s)=2$ and by (\ref{class-2})
$A_s=2$, $B_s=-2$, $C_s=D_s=0$ and hence $h_s(\varepsilon \mathbf{\Delta}^*_4)=16\varepsilon$; if $s\in \left\{{2 \brack 1},{2 \brack 0},{0 \brack 2},{0 \brack 1}\right\}$, we have $\gamma(s)=2, \delta(s)=1$ and by (\ref{class-2})
$A_s=B_s=C_s=D_s=2$ and hence $h_s(\varepsilon \mathbf{\Delta}^*_4)=0$; if $s\in \mathcal{S}_4(4,4,3)$, we have $\gamma(s)=3, \delta(s)=1$ and by (\ref{class-2}) $A_s=B_s=3$, $C_s=0$ and $D_s=2$ and hence $h_s(\varepsilon \mathbf{\Delta}^*_4)=0$. For the samples in class I, we use (\ref{class-1}) to compute the coefficients and then obtain $h_s(\varepsilon \mathbf{\Delta}^*_4)$. The computations as above yield Table \ref{table-k-4}, where the values of all $g_s(\mathcal C^*_4)$ and $h_s(\varepsilon \mathbf{\Delta}^*_4)$ are listed.

\begin{table}[h]
\centering
\caption{The values of $g_s(\mathcal C^*_4)$ and $h_s(\varepsilon \mathbf{\Delta}^*_4)$. Note that $\mathcal{S}_4(2,3)=\mathcal{S}_4(2,3,0)\cup \mathcal{S}_4(2,3,1)$.}

\vspace{0.5cm}

\begin{tabular}{|l||c|c|c|c|c|}
  \hline
\multicolumn{1}{|l||}{Class of $s$} &  \multicolumn{5}{|c|}{class I} \\ \hline
  Subclass of $s$ & $\mathcal{S}_4(1,1,1)$ & $\mathcal{S}_4(1,2,0)$ & $\mathcal{S}_4(2,2,0)$ & $\mathcal{S}_4(2,2,1)$ & $\mathcal{S}_4(2,2,2)$  \\ \hline
$g_s(\mathcal C^*_4)$ & $\frac{5}{8}$ & $\frac{1}{2}$ & \multicolumn{3}{|c|}{1}  \\ \hline
  $h_s(\varepsilon \mathbf{\Delta}^*_4)$ & $18\varepsilon$ & $-8\varepsilon$ & $-16\varepsilon$ & 0 & $16\varepsilon$ \\
  \hline
\end{tabular}

\bigskip

\begin{tabular}{|l||c|c|c|c|c|c|}
  \hline
\multicolumn{1}{|l||}{Class of $s$} &  \multicolumn{6}{|c|}{class I}\\ \hline
  Subclass of $s$ & $\mathcal{S}_4(2,3)$  & $\mathcal{S}_4(2,4,0)$ & $\mathcal{S}_4(3,3)$ & $\mathcal{S}_4(3,4,0)$ & $\mathcal{S}_4(3,4,1)$ & $\mathcal{S}_4(3,4,2)$ \\ \hline
  $g_s(\mathcal C^*_4)$ & $\frac{3}{4}$ & $\frac{1}{2}$ &  $\frac{9}{8}$ & \multicolumn{3}{|c|}{$\frac{3}{4}$} \\ \hline
  $h_s(\varepsilon \mathbf{\Delta}^*_4)$ & $-4\varepsilon$ & $8\varepsilon$ & $-6\varepsilon$ & $12\varepsilon$ & $-4\varepsilon$ & $-20\varepsilon$ \\  \hline
\end{tabular}

\bigskip

\begin{tabular}{|l||c|c|c|c|}
  \hline
\multicolumn{1}{|l||}{Class of $s$} &  \multicolumn{3}{|c|}{class I} & class II \\ \hline
  Subclass of $s$ & $\mathcal{S}_4(4,4,0)$ & $\mathcal{S}_4(4,4,1)$ & $\left\{{2 \brack 0},{2 \brack 0},{0 \brack 2},{0 \brack 2}\right\}$ & $\left\{{2 \brack 1},{2 \brack 0},{0 \brack 1},{0 \brack 2}\right\}$ \\ \hline
$g_s(\mathcal C^*_4)$  & \multicolumn{4}{|c|}{1}  \\ \hline
  $h_s(\varepsilon \mathbf{\Delta}^*_4)$ & $16\varepsilon$ & 0 & $-16\varepsilon$ & 0 \\
  \hline
\end{tabular}

\bigskip

\begin{tabular}{|l||c|c|c|c|}
  \hline
\multicolumn{1}{|l||}{Class of $s$} & class I &  \multicolumn{3}{|c|}{class II}  \\ \hline
  Subclass of $s$ & $\mathcal{S}_4(4,4,4)$ & $\left\{{2 \brack 1},{2 \brack 1},{0 \brack 1},{0 \brack 1}\right\}$ & $\left\{{2 \brack 2},{2 \brack 0},{0 \brack 1},{0 \brack 1}\right\}$ & $\mathcal{S}_4(4,4,3)$ \\ \hline
$g_s(\mathcal C^*_4)$  & \multicolumn{4}{|c|}{1}  \\ \hline
  $h_s(\varepsilon \mathbf{\Delta}^*_4)$ & $-48\varepsilon$ & \multicolumn{2}{|c|}{$16\varepsilon$}  & 0 \\
  \hline
\end{tabular}
\label{table-k-4}
\end{table}

Note that by Table \ref{table-k-4}, $\{g_{s}(\mathcal C^*_4)\|s\in \mathcal{S}_k\}$ achieves the maximum $\frac{9}{8}$ at $\mathcal{S}_4(3,3)$ and the second largest value $1$ at $\mathcal{S}_4(2,2)\cup \mathcal{S}_4(4,4)$. Now, similarly as in the case $k=3$, we will perturb $\mathcal{C}_4^*$ along the direction of $\mathbf{\Delta}_4^*$ to obtain $\mathcal{C}_4^{**}$. To this end, we again note that in the course of perturbation, the second largest value is always achieved at $\mathcal{S}_4(4, 4)$, and we thereby solve $1+16\varepsilon=\frac{9}{8}-6\varepsilon$, which yields $\varepsilon^*=\frac{1}{22\times8}=\frac{1}{176}$, and furthermore,
\begin{equation*}
\hspace{-1cm} \begin{split}
&\mathcal C^{**}_4=\mathbf{\Delta}_4^* \times\frac{1}{176}+\mathcal C^*_4\\
&=\left( \left(
    \begin{array}{cccc}
      \frac{5}{11} & \frac{2}{11} & \frac{2}{11} & \frac{2}{11} \\
      \frac{2}{11} &\frac{-1}{11} & \frac{-1}{22} & \frac{-1}{22}\\
      \frac{2}{11} & \frac{-1}{22} & \frac{-1}{11}& \frac{-1}{22} \\
      \frac{2}{11} & \frac{-1}{22} & \frac{-1}{22} & \frac{-1}{11} \\
    \end{array}
  \right),\left(
    \begin{array}{cccc}
      \frac{-1}{11} & \frac{2}{11} & \frac{-1}{22} & \frac{-1}{22} \\
      \frac{2}{11} &\frac{5}{11} & \frac{2}{11} &\frac{2}{11} \\
      \frac{-1}{22} & \frac{2}{11}& \frac{-1}{11}& \frac{-1}{22} \\
      \frac{-1}{22} & \frac{2}{11}& \frac{-1}{22} & \frac{-1}{11} \\
    \end{array}
  \right),\left(
    \begin{array}{cccc}
      \frac{-1}{11} & \frac{-1}{22} & \frac{2}{11} & \frac{-1}{22} \\
      \frac{-1}{22} &\frac{-1}{11} & \frac{2}{11} & \frac{-1}{22}\\
      \frac{2}{11} & \frac{2}{11} & \frac{5}{11}& \frac{2}{11} \\
      \frac{-1}{22} & \frac{-1}{22} & \frac{2}{11} & \frac{-1}{11} \\
    \end{array}
  \right),\left(
    \begin{array}{cccc}
      \frac{-1}{11} & \frac{-1}{22} & \frac{-1}{22} & \frac{2}{11} \\
      \frac{-1}{22} &\frac{-1}{11} & \frac{-1}{22} & \frac{2}{11}\\
      \frac{-1}{22} & \frac{-1}{22} & \frac{-1}{11}& \frac{2}{11} \\
      \frac{2}{11} & \frac{2}{11} & \frac{2}{11} & \frac{5}{11} \\
    \end{array}
  \right)\right).
\end{split}
\end{equation*}
By Lemma~\ref{valid-k-4}, $\mathbf{\Delta}_4^* \times \frac{1}{176}$ is a valid perturbation.

\textbf{Step 3: Compute $\mathcal S_4^{max}(\mathcal C^{**}_4)$.} From Table \ref{table-k-4}, it is easy to verify that $\{g_s(\mathcal C^{**}_4)\|s\in \mathcal{S}_4\}$ achieves the maximum $\frac{12}{11}$ at
$$
\mathcal{S}_4^{max}(\mathcal{C}_4^{**})=\mathcal{S}_4(2,2,2)\cup \mathcal{S}_4(3,3)\cup \mathcal{S}_4(4,4,0) \cup \left\{{2 \brack 1},{2 \brack 1},{0 \brack 1},{0 \brack 1}\right\}\cup \left\{{2 \brack 2},{2 \brack 0},{0 \brack 1},{0 \brack 1}\right\}.
$$

$\blacksquare$ For the cases $k=5, 7, 8, 9$, we only outline the major steps to derive $\mathcal C^{**}_k$ without giving all the computation details.

\textbf{Step 1: Compute $\mathbf{\Delta}^*_k$.}
\begin{itemize}
  \item For $k=5$, $\overline{x}(\mathbf{\Delta}^*_5)=4$, $\overline{a}(\mathbf{\Delta}^*_5)=-1$, $\overline{b}(\mathbf{\Delta}^*_5)=\frac{8}{7}$, $\overline{y}(\mathbf{\Delta}^*_5)=-\frac{17}{7}$.
  \item For $k=7$, $\overline{x}(\mathbf{\Delta}^*_7)=6$, $\overline{a}(\mathbf{\Delta}^*_7)=-1$, $\overline{b}(\mathbf{\Delta}^*_7)=\frac{10}{11}$, $\overline{y}(\mathbf{\Delta}^*_7)=-\frac{39}{11}$.
  \item For $k=8$, $\overline{x}(\mathbf{\Delta}^*_8)=7$, $\overline{a}(\mathbf{\Delta}^*_8)=-1$, $\overline{b}(\mathbf{\Delta}^*_8)=\frac{3}{5}$, $\overline{y}(\mathbf{\Delta}^*_8)=-\frac{13}{5}$.
  \item For $k=9$, $\overline{x}(\mathbf{\Delta}^*_9)=8$, $\overline{a}(\mathbf{\Delta}^*_9)=-1$, $\overline{b}(\mathbf{\Delta}^*_9)=\frac{14}{31}$, $\overline{y}(\mathbf{\Delta}^*_9)=-\frac{67}{31}$.
\end{itemize}

\textbf{Step 2: Compute $\mathcal C^{**}_k$.}
\begin{itemize}
  \item For $k=5$, $g_s(\mathcal C^*_k)$ achieves the maximum $\frac{28}{25}$ at $\mathcal S_5(4,4)$ and the second largest value $\frac{27}{25}$ at $\mathcal S_5(3,3)$.  By Definition~\ref{defn-class}, all these samples are of class I. Then, an application of Lemma~\ref{val-dif-sam-en} yields that
      $$
      h_s(\varepsilon\mathbf{\Delta}^*_5)=\left\{
                                     \begin{array}{ll}
                                       -\frac{24}{7}\varepsilon, & \hbox{$s\in \mathcal S_5(4,4)$,} \\
                                       \frac{50\gamma(s)-66}{7}\varepsilon, & \hbox{$s\in \mathcal S_5(3,3)$,}                                  \end{array}
                                   \right.
      $$
      based on which, we infer that the second largest value increases the fastest (with speed $h_s(\varepsilon\mathbf{\Delta}^*_5)=\frac{84\varepsilon}{7}$) when $\gamma(s)=3$. Solving the equation $\frac{28}{25}-\frac{24}{7}\varepsilon=\frac{27}{25}+\frac{84}{7}\varepsilon$, we have $\varepsilon^*=\frac{7}{108\times25}$ and obtain $\mathcal C^{**}_5=\mathcal C^*_5+\frac{7}{108\times25}\mathbf{\Delta}^*_5$ with
      $$
      \left\{
                                     \begin{array}{ll}
                                       x(\mathcal C^{**}_5)=\frac{40}{108},\\
                                       a(\mathcal C^{**}_5)=\frac{17}{108},\\
                                       b(\mathcal C^{**}_5)=\frac{-4}{108},\\
                                       y(\mathcal C^{**}_5)=\frac{-5}{108}.
                                     \end{array}
                                   \right.
      $$

  \item For $k=7$, $g_s(\mathcal C^*_k)$ achieves the maximum $\frac{55}{49}$ at $\mathcal S_7(5,5)$ and the second largest value $\frac{54}{49}$ at $\mathcal S_7(6,6)$. By Definition~\ref{defn-class}, all these samples are of class I. Then, an application of Lemma~\ref{val-dif-sam-en} yields that
      $$
      h_s(\varepsilon\mathbf{\Delta}^*_7)=\left\{
                                     \begin{array}{ll}
                                       -\frac{60}{11}\varepsilon, & \hbox{$s\in \mathcal S_7(5,5)$,} \\
                                       \frac{-98\gamma(s)+48}{11}\varepsilon, & \hbox{$s\in \mathcal S_7(6,6)$,}                                  \end{array}
                                   \right.
      $$
      based on which we infer that the second largest value increases the fastest (with speed $h_s(\varepsilon\mathbf{\Delta}^*_7)=\frac{48\varepsilon}{11}$) when $\gamma(s)=0$. Solving the equation $\frac{55}{49}-\frac{60}{11}\varepsilon=\frac{54}{49}+\frac{48}{11}\varepsilon$, we have $\varepsilon^*=\frac{11}{108\times49}$ and obtain $\mathcal C^{**}_7=\mathcal C^*_7+\frac{11}{108\times49}\mathbf{\Delta}^*_7$ with
      $$
      \left\{
                                     \begin{array}{ll}
                                       x(\mathcal C^{**}_7)=\frac{30}{108},\\
                                       a(\mathcal C^{**}_7)=\frac{13}{108},\\
                                       b(\mathcal C^{**}_7)=\frac{-2}{108},\\
                                       y(\mathcal C^{**}_7)=\frac{-3}{108}.
                                     \end{array}
                                   \right.
      $$

  \item For $k=8$, $g_s(\mathcal C^*_k)$ achieves the maximum $\frac{9}{8}$ at $\mathcal S_8(6,6)$ and the second largest value $\frac{35}{32}$ at $\mathcal S_8(5,5) \cup \mathcal S_8(7,7)$. By Definition~\ref{defn-class}, all these samples are of class I. Then, an application of Lemma~\ref{val-dif-sam-en} yields that
      $$
      h_s(\varepsilon\mathbf{\Delta}^*_8)=\left\{
                                     \begin{array}{ll}
                                       -\frac{24}{5}\varepsilon, & \hbox{$s\in \mathcal S_8(6,6)$,} \\
                                       \frac{32\gamma(s)-50}{5}\varepsilon, & \hbox{$s\in \mathcal S_8(5,5)$,}\\
                                       \frac{-32\gamma(s)+14}{11}\varepsilon, & \hbox{$s\in \mathcal S_8(7,7)$,}                                  \end{array}
                                   \right.
      $$
      based on which we infer that the second largest value increases the fastest (with speed $h_s(\varepsilon\mathbf{\Delta}^*_8)=\frac{110\varepsilon}{5}$) when $s\in \mathcal S_8(5,5,5)$. Solving the equation $\frac{9}{8}-\frac{24}{5}\varepsilon=\frac{35}{32}+\frac{110}{5}\varepsilon$, we have $\varepsilon^*=\frac{5}{134\times32}$ and obtain $\mathcal C^{**}_8=\mathcal C^*_8+\frac{5}{134\times32}\mathbf{\Delta}^*_8$ with
      $$
      \left\{
                                     \begin{array}{ll}
                                       x(\mathcal C^{**}_8)=\frac{65}{268},\\
                                       a(\mathcal C^{**}_8)=\frac{29}{268},\\
                                       b(\mathcal C^{**}_8)=\frac{-4}{268},\\
                                       y(\mathcal C^{**}_8)=\frac{-5}{268}.
                                     \end{array}
                                   \right.
      $$

  \item For $k=9$, $g_s(\mathcal C^*_k)$ achieves the maximum $\frac{91}{81}$ at $\mathcal S_9(7,7)$ and the second largest value $\frac{90}{81}$ at $\mathcal S_9(5,5)\cup\mathcal S_9(6,6)$.  By Definition~\ref{defn-class}, all these samples are of class I. Then, an application of Lemma~\ref{val-dif-sam-en} yields that
      $$
      h_s(\varepsilon\mathbf{\Delta}^*_8)=\left\{
                                       \begin{array}{ll}
                                       -\frac{140}{31}\varepsilon, & \hbox{$s\in \mathcal S_9(7,7)$,} \\
                                       \frac{168\gamma(s)-288}{31}\varepsilon, & \hbox{$s\in \mathcal S_9(6,6)$,}                              \end{array}
                                   \right.
      $$
      based on which, we infer that the second largest value increased fastest (with speed $h_s(\varepsilon\mathbf{\Delta}^*_9)=\frac{684\varepsilon}{31}$) when $s\in \mathcal S_9(6,6,6)$. Solving the equation $\frac{91}{81}-\frac{140}{31}\varepsilon=\frac{90}{81}+\frac{684}{31}\varepsilon$, we have $\varepsilon^*=\frac{31}{824\times81}$ and obtain $\mathcal C^{**}_9=\mathcal C^*_9+\frac{31}{824\times81}\mathbf{\Delta}^*_9$ with
      $$
      \left\{
                                     \begin{array}{ll}
                                       x(\mathcal C^{**}_9)=\frac{176}{824},\\
                                       a(\mathcal C^{**}_9)=\frac{81}{824},\\
                                       b(\mathcal C^{**}_9)=\frac{-10}{824},\\
                                       y(\mathcal C^{**}_9)=\frac{-11}{824}.
                                     \end{array}
                                   \right.
      $$
\end{itemize}

\textbf{Step 3: Compute $\mathcal S_k^{max}(\mathcal C^{**}_k)$.}
\begin{itemize}
  \item For $k=5$, it can be easily verified that $\varepsilon=\frac{7}{108\times 25}$ is $g_{\mathcal S_5}$-valid. Hence, we compute $g_s(\mathcal C^{**}_5)=g_s(\mathcal C^{*}_5)+h_s(\varepsilon\mathbf{\Delta}^*_5)$ for all $s\in \mathcal S_5$ using Lemma~\ref{flow-value} and Lemma~\ref{val-dif-sam-en}. It turns out $g_{\mathcal S_5}(\mathcal C^{**}_7)=\frac{28}{25}-\frac{24}{2700}=\frac{10}{9}$ is achieved at $\mathcal S_5^{max}(\mathcal C^{**}_5)=\mathcal S_5(3,3,3)\cup\mathcal S_5(4,4)$.

  \item For $k=7$, it can be easily verified that $\varepsilon=\frac{11}{108\times 49}$ is not $g_s$-valid if and only if
      \begin{equation*}
      \begin{split}
       T(s)&\in\left\{{2\brack 1},{2\brack 0},{2\brack 0},{2\brack 0},{2\brack 0},{0\brack 1},{0\brack 0}\right\}\bigcup\left\{{2\brack 0},{2\brack 0},{2\brack 0},{2\brack 0},{2\brack 0},{0\brack 1},{0\brack 1}\right\}\\
       &\bigcup\left\{{2\brack 1},{2\brack 0},{2\brack 0},{2\brack 0},{0\brack 1},{0\brack 1},{0\brack 1}\right\}\bigcup\left\{{2\brack 0},{2\brack 0},{2\brack 0},{2\brack 0},{0\brack 1},{0\brack 1},{0\brack 2}\right\}\\
       &\bigcup\left\{{2\brack 1},{2\brack 1},{2\brack 0},{2\brack 0},{0\brack 1},{0\brack 1},{0\brack 0}\right\}\bigcup\left\{{2\brack 2},{2\brack 0},{2\brack 0},{2\brack 0},{0\brack 1},{0\brack 1},{0\brack 0}\right\}\\
       &\bigcup\left\{{2\brack 1},{2\brack 0},{2\brack 0},{2\brack 0},{0\brack 1},{0\brack 2},{0\brack 0}\right\}.
      \end{split}
      \end{equation*}
      It turns out $g_s(\mathcal C^{**}_7)\leq\frac{29}{27}$ for all the samples $s$ of the types as above. We compute $g_s(\mathcal C^{**}_7)=g_s(\mathcal C^{*}_7)+h_s(\varepsilon\mathbf{\Delta}^*_7)$ for all the other samples $s$ using Lemma \ref{flow-value} and Lemma~\ref{val-dif-sam-en}. It turns out $g_{\mathcal S_7}(\mathcal C^{**}_7)=\frac{55}{49}-\frac{60}{108\times49}=\frac{10}{9}$ is achieved at $\mathcal{S}_7^{max}(\mathcal C^{**}_7)=S_7(4,4,4)\cup\mathcal S_7(5,5)\cup\mathcal S_7(6,6,0)$.

  \item For $k=8$, it can be easily verified that $\varepsilon=\frac{11}{108\times 49}$ is not $g_s$-valid if and only if
      \begin{equation*}
      \hspace{-0.5cm} \begin{split}
       T(s)&\in\left\{{2\brack 0},{2\brack 0},{2\brack 0},{2\brack 0},{2\brack 0},{2\brack 0},{0\brack 1},{0\brack 1}\right\}\bigcup\left\{{2\brack 1},{2\brack 0},{2\brack 0},{2\brack 0},{2\brack 0},{2\brack 0},{0\brack 1},{0\brack 0}\right\}.
      \end{split}
      \end{equation*}
      It turns out $g_s(\mathcal C^{**}_8)\leq\frac{71}{67}$ for all the samples $s$ of the types as above. We compute $g_s(\mathcal C^{**}_8)=g_s(\mathcal C^{*}_8)+h_s(\varepsilon\mathbf{\Delta}^*_8)$ for all the other $s\in \mathcal S_8$ using Lemma \ref{flow-value} and Lemma~\ref{val-dif-sam-en}. It turns out $g_{\mathcal S_8}(\mathcal C^{**}_8)=\frac{9}{8}-\frac{24}{134\times32}=\frac{75}{67}$ is achieved at $\mathcal S_8^{max}(\mathcal C^{**}_8)=S_8(5,5,5)\cup\mathcal S_8(6,6)$.

  \item For $k=9$, it can be easily verified that $\varepsilon=\frac{31}{824\times 81}$ is $g_{\mathcal S_9}$-valid. We compute $g_s(\mathcal C^{**}_8)=g_s(\mathcal C^{*}_8)+h_s(\varepsilon\mathbf{\Delta}^*_8)$ for all $s\in \mathcal S_8$ using Lemma \ref{flow-value} and Lemma~\ref{val-dif-sam-en}. It turns out $g_{\mathcal S_9}(\mathcal C^{**}_9)=\frac{91}{81}-\frac{140}{824\times81}=\frac{231}{206}$ is achieved at $\mathcal S_9^{max}(\mathcal C^{**}_9)=S_9(6,6,6)\cup\mathcal S_9(7,7)$.
\end{itemize}

\subsection{Optimality of $\mathcal C^{**}_k$ for $k=3, 4, 5, 7, 8, 9$} \label{section-optimality}

In this section, we prove that $\mathcal C^{**}_k$ obtained in the last section are optimal solutions to $\mathcal P_{\mathcal S_k}$ for $k=3, 4, 5, 7, 8, 9$. We first introduce more notations and state some needed lemmas.

Recall that for any sample $s \in \mathcal{S}^{\circ}_k(\mathcal{C})$, we have $g^{(\ell)}_s(\mathcal C)>0$ for any $\ell \in Ind_s$, and $g^{(\ell)}_s(\mathcal C)<0$ for any $\ell \notin Ind_s$. Since any function $g^{(\ell)}_s$, $s\in \mathcal{S}_k^{\circ}(\mathcal C)$, is continuous, there exists a neighborhood, denoted by $N(\mathcal C, \varepsilon) \subset \mathfrak S_k$, of $\mathcal C$ such that for all $\mathcal C' \in N(\mathcal C, \varepsilon)$, all $s\in \mathcal{S}_k(\mathcal C)$ and all $1 \leq \ell \leq k$,
$$
g^{(\ell)}_s(\mathcal C')\cdot g^{(\ell)}_s(\mathcal C)>0.
$$

For $\mathcal C'\in N(\mathcal C, \varepsilon)$, we write
\begin{equation} \label{Delta}
\mathbf{\Delta} := \mathcal C'-\mathcal C=\left(\mathbf{\Delta}^{(1)}, \mathbf{\Delta}^{(2)}, \dots, \mathbf{\Delta}^{(k)}\right),
\end{equation}
where each $\mathbf{\Delta}^{(\ell)}=\left( \delta^{(\ell)}_{i, j}\right)$ is a $k \times k$ matrix such that for all $\ell=1,2,\dots,k$,
\begin{equation}\label{sum-column-row-k}
\sum_{i=1}^k\delta^{(\ell)}_{i,j}=\sum_{i=1}^k\delta^{(\ell)}_{i,j}=0.
\end{equation}
And moreover, we write
\begin{equation} \label{hsDelta}
h_s(\mathbf{\Delta}):=g_s(\mathcal C')-g_s(\mathcal C).
\end{equation}

We need the following three lemmas.
\begin{lem}\label{sum-H_k(a,a,a)}
Let $k\geq3$ and $d \geq2$. If $\mathcal S_k(d, d, d) \subseteq \mathcal{S}_k^{\dagger}(\mathcal C)$, then
\begin{equation}\label{H_k(a,a,a)}
\sum_{s\in \mathcal S_k(d,d,d)}h_s(\mathbf{\Delta})=\begin{bmatrix}
    \sum_{i=1}^k\delta^{(i)}_{i,i} & \sum_{i,j:i\neq j}\delta^{(i)}_{j,j}
\end{bmatrix}
\begin{bmatrix}
    A \\  B
\end{bmatrix},
\end{equation}
where $A=\binom{k-1}{d-1}$, $B=\binom{k-2}{d-2}-\binom{k-2}{d-1}$.
\end{lem}

\begin{proof}
Note that
$$
\sum_{s\in \mathcal S_k(d,d,d)}h_s(\mathbf{\Delta})=\sum_{\ell=1}^k\sum_{i=1}^k\sum_{j=1}^k h^{(\ell)}_{i,j}\delta^{(\ell)}_{i,j},
$$
where the coefficients $h^{(\ell)}_{i,j}$ can be computed as follows. Firstly, note that for any $s\in \mathcal{S}_k(d,d,d)$ and any $(i,j)\in [k]\times[k]$ with $i\neq j$, we have $(i,j) \notin s$, and hence $h^{(\ell)}_{i,j}=0$. Secondly, for each $(i,i)\in [k]\times[k]$, there are $\binom{k-1}{d-1}$ samples of $\mathcal{S}_k(d,d,d)$ containing $(i,i)$ and hence $h^{(i)}_{i,i}=\binom{k-1}{d-1}$ for all $1\leq i\leq k$. Thirdly, noticing that for any $j\neq i$, there are $\binom{k-2}{d-2}$ samples containing $(i,i)$ and $(j,j)$, and there are $\binom{k-2}{d-1}$ samples containing $(i,i)$ but not $(j,j)$, we have $h^{(j)}_{i,i}=\binom{k-2}{d-2}-\binom{k-2}{d-1}$, which completes the proof.
\end{proof}

\begin{lem}\label{sum-H_k(a,a,0)}
Let $k\geq3$ and $d\geq3$. If $\mathcal S_k(d,d,0)\subseteq\mathcal S^{\dagger}_k(\mathcal{C})$, then
\begin{equation}\label{H_k(a,a,0)}
\sum_{s\in \mathcal S_k(d,d,0)}h_s(\mathbf{\Delta})=\begin{bmatrix}
    \sum_{i=1}^k\delta^{(i)}_{i,i} & \sum_{i,j:i\neq j}\delta^{(i)}_{j,j}
\end{bmatrix}
\begin{bmatrix}
    A' \\  B'
\end{bmatrix},
\end{equation}
where $A'=(d-1)^{(d-1)}\cdot\left(\binom{k-3}{d-2}-\binom{k-3}{d-3}-2\binom{k-2}{d-2}\right)$,
$B'=(d-1)^{(d-1)}\cdot\left(\binom{k-3}{d-2}-\binom{k-3}{d-3}\right)$.
\end{lem}

\begin{proof}
Note that
$$
\sum_{s\in \mathcal S_k(d,d,0)}h_s(\mathbf{\Delta})=\sum_{\ell=1}^k\sum_{i=1}^k\sum_{j=1}^k h^{(\ell)}_{i,j}\delta^{(\ell)}_{i,j},
$$
where the coefficients $h^{(\ell)}_{i,j}$ can be computed as follows. Firstly, note that for any $s\in \mathcal{S}_k(d,d,0)$, we have $(i,i) \notin s$, and hence $h^{(\ell)}_{i,i}=0$ for all $1\leq i,\ell\leq k$. Secondly, for each $i\neq j$, there are $(d-1)^{(d-1)}\cdot\binom{k-2}{d-2}$ samples from $\mathcal{S}_k(d,d,0)$ containing $(i,j)$ and hence $h^{(i)}_{i,j}=h^{(j)}_{i,j}=(d-1)^{(d-1)}\cdot\binom{k-2}{d-2}$. Thirdly, noticing that for any distinct $i,j,\ell$, there are $(d-1)^{(d-1)}\cdot\binom{k-3}{d-3}$ samples $s$ such that $\ell\in Ind_s$ and $(i,j)\in s$, and there are $(d-1)^{(d-1)}\cdot\binom{k-3}{d-2}$ samples $s$ such that $\ell\notin Ind_s$ and $(i,j)\in s$, we have $h^{(\ell)}_{i,j}=(d-1)^{(d-1)}\cdot\left(\binom{k-3}{d-3}-\binom{k-3}{d-2}\right)$. Finally, the desired result follows from (\ref{sum-column-row-k}).
\end{proof}

\begin{lem}\label{sum-H_3(2,2,0)}
If $\mathcal S_3(2,2,0)\subseteq\mathcal S^{\dagger}_3(\mathcal C)$, then
\begin{equation}\label{H_3(2,2,0)}
\sum_{s\in \mathcal S_3(2,2,0)}h_s(\mathbf{\Delta})=\begin{bmatrix}
    \sum_{i=1}^3\delta^{(i)}_{i,i} & \sum_{i,j:i\neq j}\delta^{(i)}_{j,j}
\end{bmatrix}
\begin{bmatrix}
    -3 \\  1
\end{bmatrix}.
\end{equation}
\end{lem}
\begin{proof}
For the 3 samples in $S_3(2,2,0)$, we have
\begin{equation*}
\begin{split}
h_{\{(2,1),(1,2)\}}(\Delta)=&(\delta^{(1)}_{2,1}+\delta^{(1)}_{1,2})+(\delta^{(2)}_{2,1}+\delta^{(2)}_{1,2})
-(\delta^{(3)}_{2,1}+\delta^{(3)}_{1,2}),\\
h_{\{(3,2),(2,3)\}}(\Delta)=&(\delta^{(2)}_{3,2}+\delta^{(2)}_{2,3})
+(\delta^{(3)}_{3,2}+\delta^{(3)}_{2,3})-(\delta^{(1)}_{3,2}+\delta^{(1)}_{2,3}),\\
h_{\{(3,1),(1,3)\}}(\Delta)=&(\delta^{(1)}_{3,1}+\delta^{(1)}_{1,3})
+(\delta^{(3)}_{3,1}+\delta^{(3)}_{1,3})-(\delta^{(2)}_{3,1}+\delta^{(2)}_{1,3}).
\end{split}
\end{equation*}
Hence,
\begin{equation}\label{sum-S_3(2,2,0)}
\begin{split}
\sum_{s\in S_3(2,2,0)}h_s(\Delta)=&\sum_{\ell=1}^3\sum_{i\neq j}\delta^{(\ell)}_{i,j}-2\sum_{\text{distinct}\;i,j,\ell }\delta^{(\ell)}_{i,j}\\
=&\sum_{i=1}^3-3\delta^{(i)}_{i,i}+\sum_{i,j:i\neq j}\delta^{(i)}_{j,j},
\end{split}
\end{equation}
which complete the proof.
\end{proof}

The following lemma gives a sufficient condition for the local optimality of an arbitrary $\mathcal C\in \mathfrak S_k$.

\begin{lem}\label{local-optimality}
Let $\mathcal C\in \mathfrak S_k$. If there exists a subset $\mathcal{S}^{\circ}_k(\mathcal C) \subseteq \mathcal{S}^{\dagger}_k(\mathcal{C})$ and a neighborhood $N(\mathcal C, \varepsilon) \subset \mathfrak S_k$ of $\mathcal C$ and a set of positive reals $\{k_s\|s\in \mathcal{S}^{\circ}_k(\mathcal C)\}$ such that for all $\mathcal C'\in N(\mathcal C, \varepsilon)$,
\begin{equation}\label{pos-lin-com-k}
\sum_{s\in \mathcal{S}^{\circ}_k(\mathcal C)}k_s\cdot (g_s(\mathcal C')-g_s(\mathcal C))=0.
\end{equation}
Then, $\mathcal C$ is a local optimal point for $\mathcal P_{\mathcal S_k}$.
\end{lem}

\begin{proof}
Suppose, by way of contradiction, that $\mathcal C$ is not a local optimal point, i.e., there exists a neighborhood $N(\mathcal C, \varepsilon)$ of $\mathcal C$ and $\mathcal C'\in N(\mathcal C, \varepsilon)$ such that $g_{\mathcal S_k}(\mathcal C')<g_{\mathcal S_k}(\mathcal C)$. Then, for all $s\in \mathcal{S}_k(\mathcal C)$,
$$
g_{s}(\mathcal C')\leq \max_{s\in \mathcal S_k}\{g_s(\mathcal C')\}=g_{\mathcal S_k}(\mathcal C')<g_{\mathcal S_k}(\mathcal C)=\max_{s\in \mathcal S_k}\{g_s(\mathcal C)\}=g_{s}(\mathcal C),
$$
which contradicts (\ref{pos-lin-com-k}). Hence, $\mathcal C$ is an local optimal point.
\end{proof}

From now on, we denote the vector $\begin{bmatrix} A\\ B\end{bmatrix}$ in ($\ref{H_k(a,a,a)}$) by $\mathcal H_k(d,d,d)$, $\frac{1}{(d-1)^{(d-1)}}\begin{bmatrix} A'\\ B'\end{bmatrix}$ by $\mathcal H_k(d,d,0)$ where $\begin{bmatrix} A'\\ B'\end{bmatrix}$ is obtained in ($\ref{H_k(a,a,0)}$) and $\mathcal H_3(2,2,0)=\begin{bmatrix} -3\\ 1\end{bmatrix}$ by (\ref{sum-S_3(2,2,0)}). We are then ready to give the main result of this section.
\begin{thm}\label{optimal-theorem}
$\mathcal C^{**}_k$ is an optimal point for $\mathcal P_{\mathcal S_k}$ for $k=3,4,5,7,8,9$.
\end{thm}

\begin{proof}
By Lemma~\ref{convex}, it suffices to prove that $\mathcal C^{**}_k$ is a local optimal point. To this end, by Lemma~\ref{local-optimality}, we only need to find a neighborhood of $\mathcal C^{**}_k$, a subset $\mathcal{S}^{\circ}_k(\mathcal C^{**}_k)$ of $\mathcal{S}^{\dagger}_k(\mathcal C^{**}_k)$ and a set of positives reals satisfying (\ref{pos-lin-com-k}). In the following, we take $N(\mathcal C^{**}_k, \varepsilon)$ as the neighborhood of $\mathcal C^{**}_k$ for each $k=3,4,5,7,8,9$.
\begin{itemize}
  \item For the case $k=3$, let $\mathcal{S}^{\circ}_3(\mathcal C^{**}_3)=\mathcal S_3(2,2,0)\cup\mathcal S_3(2,2,2)\cup\mathcal S_3(3,3,0)$. It can be verified that $\mathcal{S}^{\circ}_3(\mathcal C^{**}_3) \subseteq \mathcal S_3^{\dagger}(\mathcal C^{**}_3)$. Then, by Lemmas~\ref{sum-H_3(2,2,0)},~\ref{sum-H_k(a,a,a)} and~\ref{sum-H_k(a,a,0)}, we infer that (\ref{H_3(2,2,0)}), (\ref{H_k(a,a,a)}) and (\ref{H_k(a,a,0)}) hold with $\mathcal H_3(2,2,0)=\begin{bmatrix} -3\\ 1\end{bmatrix}$, $\mathcal H_3(2,2,2)=\begin{bmatrix} 2\\ 0\end{bmatrix}$ and $\mathcal H_3(3,3,0)=\begin{bmatrix} -1\\ -1\end{bmatrix}$, respectively. Since $\mathcal H_3(2,2,0)+2\mathcal H_3(2,2,2)+\mathcal H_3(3,3,0)=0$, an application of Lemma~\ref{local-optimality} yields that $\mathcal C^{**}_3$ is an optimal point for $\mathcal P_{\mathcal S_3}$.

  \item For the case $k=4$, let $\mathcal{S}^{\circ}_4(\mathcal C^{**}_4)=\mathcal S_4(2,2,2)\cup\mathcal S_4(3,3,0)\cup\mathcal S_4(3,3,3)$, which can be verified to be a subset of $\mathcal{S}_4^{\dagger}(\mathcal C^{**}_4)$. Then, as in previous case, the desired optimality of $\mathcal C^{**}_4$ then follows from Lemmas~\ref{sum-H_k(a,a,a)}, ~\ref{sum-H_k(a,a,0)} and~\ref{local-optimality} and the easily verifiable fact that $2\mathcal H_4(2,2,2)+3\mathcal H_4(3,3,0)+2\mathcal H_4(3,3,3)=0$, where $\mathcal H_4(2,2,2)=\begin{bmatrix} 3\\ -1\end{bmatrix}$, $\mathcal H_4(3,3,0)=\begin{bmatrix} -4\\ 0\end{bmatrix}$ and $\mathcal H_4(3,3,3)=\begin{bmatrix} 3\\ 1\end{bmatrix}$.

  \item For the case $k=5$, let $\mathcal{S}^{\circ}_5(\mathcal C^{**}_5)=\mathcal S_5(3,3,3)\cup\mathcal S_5(4,4,0)\cup\mathcal S_5(4,4,4)\subseteq\mathcal S_5^{\dagger}(\mathcal C^{**}_5)$. Then, the desired optimality of $\mathcal C^{**}_4$ then follows from Lemmas~\ref{sum-H_k(a,a,a)}, ~\ref{sum-H_k(a,a,0)} and~\ref{local-optimality} and the easily verifiable fact that $\mathcal H_5(3,3,3)+2\mathcal H_5(4,4,0)+\mathcal H_5(4,4,4)=0$.
    
  \item For the case $k=7$, let $\mathcal{S}^{\circ}_7(\mathcal C^{**}_7)=\mathcal S_7(4,4,4)\cup\mathcal S_7(5,5,0)\cup\mathcal S_7(5,5,5)\subseteq\mathcal S_7^{\dagger}(\mathcal C^{**}_7)$. Then, the desired optimality of $\mathcal C^{**}_4$ then follows from the fact that $3\mathcal H_7(4,4,4)+5\mathcal H_7(5,5,0)+2\mathcal H_7(5,5,5)=0$.
                 
  \item For the case $k=8$, let $\mathcal{S}^{\circ}_8(\mathcal C^{**}_8)=\mathcal S_8(5,5,5)\cup\mathcal S_8(6,6,0)\cup\mathcal S_8(6,6,6)\subseteq\mathcal S_8^{\dagger}(\mathcal C^{**}_8)$. Then, the desired optimality then follows from the fact that $12\mathcal H_8(5,5,5)+21\mathcal H_8(6,6,0)+5\mathcal H_8(6,6,6)=0$.
      
  \item For the case $k=9$, let $\mathcal{S}^{\circ}_9(\mathcal C^{**}_9)=\mathcal S_9(6,6,6)\cup\mathcal S_9(7,7,0)\cup\mathcal S_9(7,7,7)\subseteq\mathcal S_9^{\dagger}(\mathcal C^{**}_9)$. Then, the desired optimality then follows the fact that $15\mathcal H_9(6,6,6)+28\mathcal H_9(7,7,0)+3\mathcal H_9(7,7,7)=0$.
\end{itemize}
\end{proof}

\subsection{The Uniqueness of Optimal Solutions for $k=3, 4, 5, 7, 8, 9$}\label{section-unique}

We are concerned with the uniqueness of the optimal solutions to $\mathcal{P}_{\mathcal{S}_k}$. Note that the case of $k=1$ is trivial, and it is known from the proof of Theorem $3$ of \cite{CH17} that $\mathcal C^*_2$ is the unique optimal point for $\mathcal P_{\mathcal S_2}$. In this section, we will show that the optimal solutions to $\mathcal P_{\mathcal S_k}$ are unique for $k=3, 4, \dots, 9$, which however ceases to hold true for $k=10$.

We first need the following lemma, which strengthens Lemma~\ref{local-optimality}.
\begin{lem}\label{unique-local-optimal}
Let $\mathcal C\in \mathfrak S_k$. If there exists a subset $\mathcal{S}^{\star}_k(\mathcal C)\subseteq \mathcal{S}_k^{\dagger}(\mathcal C)$ and a neighborhood $N(\mathcal C, \varepsilon) \subset \mathfrak S_k$ of $\mathcal C$ and a set of positive reals $\{k_s\|s\in \mathcal{S}^{\star}_k(\mathcal C)\}$ such that (1) For all $\mathcal C'\in N(\mathcal C, \varepsilon)$, $\sum_{s\in \mathcal{S}^{\star}_k(\mathcal C)}k_s\cdot h_s(\mathbf{\Delta})=0$; (2) If for all $s\in \mathcal S^{\star}_k(\mathcal C)$, $h_s(\mathbf{\Delta})=0$ then $\mathbf{\Delta}=0$, where, as before, $\mathbf{\Delta}=\mathcal{C}'-\mathcal{C}$ and $h_s(\mathbf{\Delta})=g_s(\mathcal C')-g_s(\mathcal C)$. Then, $\mathcal C$ is the unique local optimal point for $\mathcal P_{\mathcal S_k}$.
\end{lem}

\begin{proof}
Suppose, by way of contradiction, that there exists another optimal point $\mathcal C'\in N(\mathcal C, \varepsilon)$ such that $\mathcal C'-\mathcal C=\mathbf{\Delta}\neq 0$. By Condition $(1)$, we know that for all $s\in \mathcal{S}^{\star}_k(\mathcal C)$, $g_s(\mathcal C')=g_s(\mathcal C)$, i.e., $h_s(\mathbf{\Delta})=0$ (Since otherwise there exist $s_0, s_1\in \mathcal S^{\star}_k(\mathcal C)$ such that $g_{s_0}(\mathcal C')-g_{s_0}(\mathcal C)>0$ and $g_{s_1}(\mathcal C')-g_{s_1}(\mathcal C)<0$, which contradicts the optimality of $\mathcal C$). Hence, by Condition $(2)$, we have $\mathbf{\Delta}=0$, which contradicts the assumption that $\mathbf{\Delta}\neq0$ and thereby the result follows.
\end{proof}

In the following, we set
$$
\mathcal{S}_3^{\star}(\mathcal C^{**}_3):=\mathcal{S}_3(2,2)\cup \mathcal{S}_3(3,3,0).
$$
And it can be easily verified that $\mathcal{S}_3^{\star}(\mathcal C^{**}_3) \subseteq \mathcal S^{\dagger}_3(\mathcal C^{**}_3)$. The following lemma can be used to establish the uniqueness of $\mathcal C^{**}_3$ for $\mathcal P_{\mathcal S_3}$. 
\begin{lem}\label{balance-k-3}
There exist a neighborhood $N(\mathcal C^{**}_3, \varepsilon) \subset \mathfrak{S}_3$ of $\mathcal C^{**}_3$ and a set of positive reals $\{k_s\|s\in \mathcal{S}_3^{\star}(\mathcal C^{**}_3)\}$ such that for all $\mathcal{C}' \in N(\mathcal C^{**}_3, \varepsilon)$
\begin{equation}\label{pos-lin-com-3}
\sum_{s\in \mathcal{S}_3^{\star}(\mathcal C^{**}_3)}k_s\cdot (g_s(\mathcal{C}')-g_s(\mathcal C^{**}_3))=0.
\end{equation}
\end{lem}

\begin{proof}
From the proof of Theorem \ref{optimal-theorem}, we have that for all $\mathcal{C}' \in N(\mathcal C^{**}_3, \varepsilon)$,
$$
\sum_{s\in \mathcal{S}_3(2,2,0)} h_s(\mathbf{\Delta})+2\sum_{s\in \mathcal{S}_3(2,2,2)}h_s(\mathbf{\Delta})+\sum_{s\in \mathcal{S}_3(3,3,0)}h_s(\mathbf{\Delta})=0,
$$
where, as before, $\mathbf{\Delta}=\mathcal{C}'-\mathcal{C}_3^{**}$. Since
\begin{equation*}
\begin{split}
h_{\{(1,1),(1,2)\}}(\mathbf{\Delta})+h_{\{(2,1),(2,2)\}}(\mathbf{\Delta})=&h_{\{(1,1),(2,2)\}}(\mathbf{\Delta})+h_{\{(2,1),(1,2)\}}(\mathbf{\Delta}),\\
h_{\{(1,1),(1,3)\}}(\mathbf{\Delta})+h_{\{(3,1),(3,3)\}}(\mathbf{\Delta})=&h_{\{(1,1),(3,3)\}}(\mathbf{\Delta})+h_{\{(3,1),(1,3)\}}(\mathbf{\Delta}),\\
h_{\{(2,2),(2,3)\}}(\mathbf{\Delta})+h_{\{(3,2),(3,3)\}}(\mathbf{\Delta})=&h_{\{(2,2),(3,3)\}}(\mathbf{\Delta})+h_{\{(3,2),(2,3)\}}(\mathbf{\Delta}),
\end{split}
\end{equation*}
we have
\begin{equation} \label{sum-S_3(2,2,1)}
\sum_{s\in \mathcal{S}_3(2,2,1)} h_s(\mathbf{\Delta})=\sum_{s \in \mathcal{S}_3(2,2,0)} h_s(\mathbf{\Delta})+\sum_{s\in \mathcal{S}_3(2,2,2)} h_s(\mathbf{\Delta}).
\end{equation}
Hence, we have
$$
\frac{1}{2}\sum_{s\in \mathcal{S}_3(2,2,0)} h_s(\mathbf{\Delta})+\frac{1}{2}\sum_{s\in \mathcal{S}_3(2,2,1)}h_s(\mathbf{\Delta})+\frac{3}{2}\sum_{s\in \mathcal{S}_3(2,2,2)}h_s(\mathbf{\Delta})+\sum_{s\in \mathcal{S}_3(3,3,0)}h_s(\mathbf{\Delta})=0.
$$
The proof is then complete.
\end{proof}

In the following, we set
$$
\mathcal{S}^{\star}_4(\mathcal C^{**}_4):=\mathcal{S}_4(3,3,0)\cup \mathcal{S}_4(3,3,1)\cup \mathcal{S}_4(4,4,0).
$$
It can be easily verified that $\mathcal{S}^{\star}_4(\mathcal C^{**}_4)\subseteq \mathcal{S}^{\dagger}_4(\mathcal C^{**}_4)$. The following lemma, whose proof has been postponed to Appendix~\ref{proof-balance-k-4}, can be used to establish the uniqueness of $\mathcal C^{**}_4$ for $\mathcal P_{\mathcal S_4}$.
\begin{lem}\label{balance-k-4}
There exists a neighborhood $N(\mathcal C^{**}_4, \varepsilon) \subset \mathfrak S_4$ of $\mathcal C^{**}_4$ and a set of positive reals $\{k_s\|s\in \mathcal{S}^{\star}_4(\mathcal C^{**}_4)\}$ such that for all $\mathcal{C}' \in N(\mathcal C^{**}_4, \varepsilon)$,
\begin{equation}\label{pos-lin-com-4}
\sum_{s\in \mathcal{S}^{\star}_4(\mathcal C^{**}_4)}k_s\cdot (g_s(\mathcal{C}')-g_s(\mathcal C^{**}_4))=0.
\end{equation}
\end{lem}

We are now ready to give the main result of this section.
\begin{thm}
$\mathcal C^{**}_k$ is the unique optimal point for $\mathcal P_{\mathcal{S}_k}$ for $k=3, 4, \dots, 9$.
\end{thm}

\begin{proof}

$\blacksquare$ We first deal with the case $k=3$. By Lemma~\ref{unique-local-optimal} and then Lemma~\ref{balance-k-3}, we only need to prove that the equation
\begin{equation}\label{Eq-balance-k-3}
h_s(\mathbf{\Delta})=0 \quad \mbox{ for all } s \in \mathcal{S}_3(\mathcal C^{**}_3)
\end{equation}
has the unique solution $\mathbf{\Delta}=0$.

Suppose $\mathbf{\Delta}=\left((\delta_{i, j}^{(1)}), (\delta_{i, j}^{(2)}), (\delta_{i, j}^{(3)}) \right)$ is a solution of (\ref{Eq-balance-k-3}). We first prove that for all $1\leq i, j\leq 3$,
\begin{equation}\label{all-3}
\pi_{i,j}:=\sum_{\ell=1}^3\delta^{(\ell)}_{i,j}=0.
\end{equation}

By (\ref{sum-column-row-k}), we have $\pi_{1,1}+\pi_{1,2}+\pi_{1,3}=0$ and then by $h_{\{(2,1),(1,2),(1,3)\}}(\mathbf{\Delta})=\pi_{2,1}+\pi_{1,2}+\pi_{1,3}=0$, we have $\pi_{1,1}=\pi_{2,1}$. Similarly, we have $\pi_{1,1}=\pi_{2,1}=\pi_{3,1}$. By (\ref{sum-column-row-k}), we have $\pi_{1,1}+\pi_{2,1}+\pi_{3,1}=0$, and hence, $\pi_{1,1}=\pi_{2,1}=\pi_{3,1}=0$. Further, in the same way, we have $\pi_{1,2}=\pi_{2,2}=\pi_{3,2}=0$ and finally we obtain (\ref{all-3}).

By $h_{\{(1,1),(1,2)\}}(\mathbf{\Delta})=0$ and (\ref{sum-column-row-k}), we have $\delta^{(3)}_{1,3}=\delta^{(1)}_{1,3}+\delta^{(2)}_{1,3}$. Hence, $0=\pi_{1,3}=2\delta^{(3)}_{1,3}$. Similarly, we can have $\delta^{(3)}_{1,3}=\delta^{(3)}_{2,3}=0$. Hence, by (\ref{sum-column-row-k}), we further have $\delta^{(3)}_{3,3}=0$. In a similar fashion, we finally have
\begin{equation}\label{i-j-i-3}
\delta^{(j)}_{i,j}=0, \quad 1\leq i,j\leq 3.
\end{equation}

Letting $\delta^{(1)}_{1,2}=a$, $\delta^{(1)}_{2,2}=b$, $\delta^{(1)}_{3,2}=c$ and using Equations (\ref{sum-column-row-k}), (\ref{all-3}) and (\ref{i-j-i-3}), we have
$$\mathbf{\Delta}=\left( \left(
    \begin{array}{ccc}
      0 & a & -a \\
      0 & b & -b \\
      0 & c & -c \\
    \end{array}
  \right),\left(
    \begin{array}{ccc}
      -a & 0 & a \\
      -b & 0 & b \\
      -c & 0 & c \\
    \end{array}
  \right),\left(
    \begin{array}{ccc}
      a & -a & 0 \\
      b & -b & 0 \\
      c & -c & 0 \\
    \end{array}
  \right)\right).$$
By $h_{\{(1,1),(2,2)\}}(\mathbf{\Delta})=0$, we have $b-a=a-b$, i.e., $a=b$. By $h_{\{(1,1),(3,3)\}}(\mathbf{\Delta})=0$, we have $c-a=a-c$, i.e., $a=c$. Hence $a=b=c$ and then by Equation (\ref{sum-column-row-k}), we have $a=b=c=0$, which means $\mathbf{\Delta}=0$. The proof is then complete.

$\blacksquare$ We now deal with the case $k=4$. By Lemma~\ref{unique-local-optimal} and then Lemma~\ref{balance-k-4}, we only need to prove that the equation
\begin{equation}\label{Eq-balance-k-4}
h_s(\mathbf{\Delta})=0 \quad \mbox{ for all } s \in \mathcal S^{\star}_4(\mathcal C^{**}_4)
\end{equation}
has the unique solution $\mathbf{\Delta}=0$.

Suppose
$\mathbf{\Delta}=\left((\delta^{(1)}_{i,j}), (\delta^{(2)}_{i,j}), (\delta^{(3)}_{i,j}), (\delta^{(4)}_{i,j})\right)$ is a solution of (\ref{Eq-balance-k-4}). We first prove that for all $i, j=1, 2, 3, 4$,
\begin{equation}\label{all-4}
\pi_{i,j}:=\sum_{\ell=1}^4\delta^{(\ell)}_{i,j}=0.
\end{equation}
By (\ref{sum-column-row-k}), we have $\pi_{1,1}+\pi_{1,2}+\pi_{1,3}+\pi_{1,4}=0$ and then by $h_{\{(2,1),(1,2),(1,3),(1,4)\}}(\mathbf{\Delta})=\pi_{2,1}+\pi_{1,2}+\pi_{1,3}+\pi_{1,4}=0$, we have $\pi_{1,1}=\pi_{2,1}$. Similarly, we have $\pi_{1,1}=\pi_{2,1}=\pi_{3,1}=\pi_{4,1}$. By (\ref{sum-column-row-k}), we have $\pi_{1,1}+\pi_{2,1}+\pi_{3,1}+\pi_{4,1}=0$, and hence, $\pi_{1,1}=\pi_{2,1}=\pi_{3,1}=\pi_{4,1}=0$. Further, in the same way, we have $\pi_{1,2}=\pi_{2,2}=\pi_{3,2}=\pi_{4,2}=0$ and finally, we can obtain (\ref{all-4}).

By $h_{\{(1,1),(1,2),(1,3)\}}(\mathbf{\Delta})=0$ and (\ref{sum-column-row-k}), we have $\delta^{(4)}_{1,4}=\delta^{(1)}_{1,4}+\delta^{(2)}_{1,4}+\delta^{(3)}_{1,4}$. Hence, $0=\pi_{1,4}=2\delta^{(4)}_{1,4}$. Similarly, we can have $\delta^{(4)}_{2,4}=\delta^{(4)}_{3,4}=0$. Hence, by (\ref{sum-column-row-k}), we further have $\delta^{(4)}_{4,4}=0$. Similarly, we can have
\begin{equation}\label{i-j-i-4}
\delta^{(j)}_{i,j}=0, \quad 1\leq i,j\leq 4.
\end{equation}

Since $h_{\{(1,1),(1,2),(1,3)\}}(\mathbf{\Delta})=h_{\{(2,1),(1,2),(1,3)\}}(\mathbf{\Delta})=0$, we have $\delta^{(1)}_{1,1}+\delta^{(2)}_{1,1}+\delta^{(3)}_{1,1}-\delta^{(4)}_{1,1}
=\delta^{(1)}_{2,1}+\delta^{(2)}_{2,1}+\delta^{(3)}_{2,1}-\delta^{(4)}_{2,1}$, and furthermore, by (\ref{all-4}), $-2\delta^{(4)}_{1,1}=-2\delta^{(4)}_{2,1}$. Similarly, by $h_{\{(1,1),(1,2),(1,3)\}}(\mathbf{\Delta})=h_{\{(3,1),(1,2),(1,3)\}}(\mathbf{\Delta})=0$,
we have
$$
\delta^{(4)}_{1,1}=\delta^{(4)}_{2,1}=\delta^{(4)}_{3,1}.
$$
Since $h_{\{(1,1),(1,2),(1,4)\}}(\mathbf{\Delta})=h_{\{(2,1),(1,2),(1,4)\}}(\mathbf{\Delta})=0$, we have $\delta^{(1)}_{1,1}+\delta^{(2)}_{1,1}+\delta^{(4)}_{1,1}-\delta^{(3)}_{1,1}
=\delta^{(1)}_{2,1}+\delta^{(2)}_{2,1}+\delta^{(4)}_{2,1}-\delta^{(3)}_{2,1}$, and furthermore, by (\ref{all-4}), $-2\delta^{(3)}_{1,1}=-2\delta^{(3)}_{2,1}$. Similarly, by $h_{\{(1,1),(1,2),(1,4)\}}(\mathbf{\Delta})=h_{\{(4,1),(1,2),(1,4)\}}(\mathbf{\Delta})=0$,
we can have
$$
\delta^{(3)}_{1,1}=\delta^{(3)}_{2,1}=\delta^{(3)}_{4,1}.
$$
Since $h_{\{(1,1),(1,3),(1,4)\}}(\mathbf{\Delta})=h_{\{(3,1),(1,3),(1,4)\}}(\mathbf{\Delta})=0$, we have $\delta^{(1)}_{1,1}+\delta^{(3)}_{1,1}+\delta^{(4)}_{1,1}-\delta^{(2)}_{1,1}
=\delta^{(1)}_{3,1}+\delta^{(3)}_{3,1}+\delta^{(4)}_{3,1}-\delta^{(2)}_{3,1}$, i.e., $-2\delta^{(2)}_{1,1}=-2\delta^{(2)}_{3,1}$ by (\ref{all-4}). Similarly, by $h_{\{(1,1),(1,3),(1,4)\}}(\mathbf{\Delta})=h_{\{(4,1),(1,3),(1,4)\}}(\mathbf{\Delta})=0$,
we can have
$$
\delta^{(2)}_{1,1}=\delta^{(2)}_{3,1}=\delta^{(2)}_{4,1}.
$$

Now, note that $\pi_{1,1}=\delta^{(1)}_{1,1}+\delta^{(2)}_{1,1}+\delta^{(3)}_{1,1}+\delta^{(4)}_{1,1}
=\delta^{(1)}_{2,1}+\delta^{(2)}_{2,1}+\delta^{(3)}_{2,1}+\delta^{(4)}_{2,1}=\pi_{2,1}$, by (\ref{i-j-i-4}) and the above discussions, we have $\delta^{(2)}_{2,1}=\delta^{(2)}_{1,1}$. Similarly, by $\pi_{1,1}=\pi_{3,1},$  we have $\delta^{(3)}_{3,1}=\delta^{(3)}_{1,1}$ and by $\pi_{1,1}=\pi_{4,1},$  we have $\delta^{(4)}_{4,1}=\delta^{(4)}_{1,1}$.  Hence, by (\ref{sum-column-row-k}), we deduce that for $i=1, 2, 3, 4$,
$$
\delta^{(i)}_{1,1}=\delta^{(i)}_{2,1}=\delta^{(i)}_{3,1}=\delta^{(i)}_{4,1}=0.
$$
Similarly, one can have that for $i=1, 2, 3, 4$, $j=2,3,4$,
$$
\delta^{(i)}_{1,j}=\delta^{(i)}_{2,j}=\delta^{(i)}_{3,j}=\delta^{(i)}_{4,j}=0.
$$
Collecting all the results above, we conclude that $\mathbf{\Delta}=0$, as desired.

$\blacksquare$ The uniqueness of the optimal solutions for $k = 5, 6, \dots, 9$ follows from a more complex yet completely parallel argument as for $k=3, 4$, and therefore we omit the details.
\end{proof}

\begin{thm}
There are at least two optimal points for $\mathcal P_{\mathcal S_{10}}$.
\end{thm}

\begin{proof}
It suffices to find an optimal point for $\mathcal P_{\mathcal S_{10}}$ that is different from $\mathcal{C}_{10}^{**}$. To this end, consider the system (\ref{Eq-2-system-simp-1}) and replace ``$<$'' and ``$>$'' by ``$\leq$'' and ``$\geq$'', respectively. Then, we have
\begin{equation}\label{Eq-non-unique}
\left\{
  \begin{array}{ll}
    \overline{a}+\overline{b}\leq 0, \\
    \overline{a}+2\overline{b}\leq 0, \\
    \overline{a}+2\overline{b}\geq 0,\\
    7\overline{a}+20\overline{b}\geq 0.\\
    \end{array}
\right.
\end{equation}
Note that the above system has solution
\begin{equation}\label{Eq-non-unique}
\left\{
  \begin{array}{ll}
    \overline{b}\geq0, \\
    \overline{a}+2\overline{b}= 0. \\
  \end{array}
\right.
\end{equation}
Now choosing $\delta > 0$ small enough and setting $b=\delta$, we obtain an optimal point $\mathcal C=\mathcal C^*_{10}+\mathbf{\Delta}$ different from $\mathcal{C}_{10}^{**}$ with $\overline{a}(\mathbf{\Delta})=-2\delta$, $\overline{b}(\mathbf{\Delta})=\delta$, $\overline{y}(\mathbf{\Delta})=-6\delta$, $\overline{x}(\mathbf{\Delta})=1+18\delta$.
\end{proof}

\subsection{Routing Rate} \label{routing-rate}

By Theorem~\ref{basic lemma}, the optimal solution $\mathcal{C}_{k}^{**}$, $k=3, 4, \dots, 10$, gives an explicit construction of multi-flows for the corresponding $k$-pair strongly reachable network. More precisely, translating the results in this section, we have constructed multi-flows of rate $(\frac{11}{12}, \dots, \frac{11}{12})$ for $k=3, 4$, rate $(\frac{9}{10}, \dots, \frac{9}{10})$ for $k=5, 6, 7$, rate $(\frac{67}{75}, \dots, \frac{67}{75})$ for $k=8$, rate $(\frac{206}{231}, \dots, \frac{206}{231})$ for $k=9$, rate $(\frac{25}{28}, \dots, \frac{25}{28})$ for $k=10$, each of which further gives a lower bound on the corresponding $\mathbf{R}_r(\overline{\mathcal{N}})$. To the best of our knowledge, the aforementioned rates are the largest to date.

\section{Concluding Remarks} \label{section-conclusion}

We attack the Langberg-M\'{e}dard multiple unicast conjecture via an optimization approach. For a closely related optimization problem $\mathcal P_{\mathcal S_k}$ with optimal value $\mathcal{O}_{\mathcal{S}_k}$, we analyze the asymptotics of $\{\mathcal{O}_{\mathcal{S}_k}\}$ and explicit solve $\mathcal P_{\mathcal S_k}$ for $k=1, 2, \dots, 10$. More precisely, we prove that $\lim_{k \to \infty } \mathcal{O}_{\mathcal{S}_k} = 9/8$, and establish the first $10$ terms of $\{\mathcal{O}_{\mathcal S_k}\}$ $1,1,\frac{12}{11}, \frac{12}{11}, \frac{10}{9}, \frac{10}{9}, \frac{10}{9}, \frac{75}{67}, \frac{231}{206}, \frac{28}{25}$, which give the largest feasible routing rate to date for the corresponding strongly reachable networks.

For any $k \neq 1, 2, 6, 10$, there exists a perturbation promising to give better solutions than $\mathcal{C}_k^*$, a sequence of asymptotically optimal solutions to $\mathcal{P}_{\mathcal{S}_k}$, and moreover, a delicate perturbation analysis in Sections~\ref{section-perturbation} and~\ref{section-perturb} gives the exact optimal solutions for $k \leq 10$. Nevertheless, it remains to be seen whether the perturbation approach can be used to solve $\mathcal{P}_{\mathcal{S}_k}$ for all $k$. The major hurdle for the case of larger $k$ is the drastically increasing complexity needed for the analysis, which is already prohibitive for $k=11$. Here we remark that the optimization problem appears to be ``trickier'' than previously thought. For a quick example, one would be tempted to think that the sequence $\{\mathcal{O}_{\mathcal{S}_k}\}$ should be monotonically increasingly. This, however, is not true, since our results actually indicate that $\mathcal{O}_{\mathcal S_9}>\mathcal{O}_{\mathcal S_{10}}$.

\printindex

\section*{Appendices} \appendix

\section{Proof of Lemma~\ref{valid-k-3}} \label{proof-valid-k-3}

By the definition of $\mathcal C^*_k$, it can be readily verified that for any $3$-sample $s$ and any $1\leq\ell\leq3$, \begin{equation}\label{g^l_s}
g^{(\ell)}_s(\mathcal C^*_3)=\frac{3}{9}m_{Ind_s}(\ell)-\frac{1}{9}\alpha(s).
\end{equation}
And note that for all $s\in \mathcal S_3$, $0\leq m_{Ind_s}(\ell)\leq4$ and
$$
h_{s}^{(\ell)}(\varepsilon \mathbf{\Delta}_k^*)=\varepsilon\cdot\sum_{(i,j)\in s}\delta^{(\ell)}_{i,j},
$$
where $\delta^{(\ell)}_{i,j}$ is defined as in $\mathbf{\Delta}_k^*=((\delta^{(1)}_{i,j}), (\delta^{(2)}_{i,j}), \dots, (\delta^{(k)}_{i,j}))$.

By Definition~\ref{valid-epsilon}, any $\varepsilon>0$ is $g^{(1)}_s$-valid since $g^{(1)}_s(\mathcal C^*_3)=0$. Recall from Example~\ref{perturbation-base-k} that $\overline{x}(\mathbf{\Delta}_3^*)=-2$, $\overline{y}(\mathbf{\Delta}_3^*)=-5$, $\overline{a}(\mathbf{\Delta}_3^*)=1$ and $\overline{b}(\mathbf{\Delta}_3^*)=4$. Then, by definition, it can be easily verified that
$$
\{g^{(\ell)}_s(\mathcal C^*_3)\|s\in \mathcal{S}_3, 1\leq\ell\leq3\} \subseteq \left\{\frac{-2}{9},\frac{-1}{9},0, \frac{1}{9},\dots, \frac{8}{9},1\right\}.
$$
We now deal with the following cases:
\begin{itemize}
  \item If $g^{(\ell)}_s(\mathcal C^*_3)=\frac{-2}{9}$, which implies $m_{Ind_s}(\ell)=0$ and $\alpha(s)=2$ (see Equation (\ref{g^l_s})), then $\varepsilon$ is $g^{(\ell)}_s$-valid if and only if
$$
\left\{
  \begin{array}{ll}
    2\overline{y}(\mathbf{\Delta}_3^*)\varepsilon\leq\frac{2}{9},\\
    2\overline{b}(\mathbf{\Delta}_3^*)\varepsilon\leq\frac{2}{9},\\
    (\overline{y}(\mathbf{\Delta}_3^*)+\overline{b}(\mathbf{\Delta}_3^*))\varepsilon\leq\frac{2}{9}.
  \end{array}
\right.
$$
Recalling from Example~\ref{perturbation-base-k} that $\overline{x}(\mathbf{\Delta}_3^*)=-2$, $\overline{y}(\mathbf{\Delta}_3^*)=-5$, $\overline{a}(\mathbf{\Delta}_3^*)=1$ and $\overline{b}(\mathbf{\Delta}_3^*)=4$, we deduce that for this case $\varepsilon>0$ is $g^{(\ell)}_s$-valid if and only if $\varepsilon\leq\frac{1}{36}$.

  \item If $g^{(\ell)}_s(\mathcal C^*_3)=\frac{-1}{9}$, which implies $m_{Ind_s}(\ell)=0$ and $\alpha(s)=1$, then $\varepsilon$ is $g^{(\ell)}_s$-valid if and only if
$$
\left\{
  \begin{array}{ll}
    \overline{y}(\mathbf{\Delta}_3^*)\varepsilon\leq\frac{1}{9},\\
    \overline{b}(\mathbf{\Delta}_3^*)\varepsilon\leq\frac{1}{9}.
   \end{array}
\right.
$$
Similarly, we deduce that for this case $\varepsilon>0$ is $g^{(\ell)}_s$-valid if and only if $\varepsilon\leq\frac{1}{36}$.

\item If $g^{(\ell)}_s(\mathcal C^*_3)=0$, then, by definition, any $\varepsilon>0$ is $g^{(\ell)}_s$-valid.

\item If $g^{(\ell)}_s(\mathcal C^*_3)=\frac{1}{9}$, which implies $m_{Ind_s}(\ell)=1$ and $\alpha(s)=2$, then $\varepsilon$ is $g^{(\ell)}_s$-valid if and only if
$$
\left\{
  \begin{array}{ll}
    -(\overline{a}(\mathbf{\Delta}_3^*)+\overline{y}(\mathbf{\Delta}_3^*))\varepsilon\leq\frac{1}{9},\\
    -(\overline{a}(\mathbf{\Delta}_3^*)+\overline{b}(\mathbf{\Delta}_3^*))\varepsilon\leq\frac{1}{9}.
   \end{array}
\right.
$$
Straightforward computations yield that that, for this case, $\varepsilon>0$ is $g^{(\ell)}_s$-valid if and only if $\varepsilon\leq\frac{1}{36}$.

\item If $g^{(\ell)}_s(\mathcal C^*_3)=\frac{2}{9}$, which implies $m_{Ind_s}(\ell)=1$ and $\alpha(s)=1$, then $\varepsilon$ is $g^{(\ell)}_s$-valid if and only if $-\overline{a}(\mathbf{\Delta}_3^*)\varepsilon\leq\frac{2}{9}$, i.e.,
any $\varepsilon>0$ is $g^{(\ell)}_s$-valid.

\item If $g^{(\ell)}_s(\mathcal C^*_3)=\frac{3}{9}$, which implies $m_{Ind_s}(\ell)=2$ and $\alpha(s)=3$, then $\varepsilon$ is $g^{(\ell)}_s$-valid if and only if
$$
\left\{
  \begin{array}{ll}
    -(2\overline{a}(\mathbf{\Delta}_3^*)+\overline{y}(\mathbf{\Delta}_3^*))\varepsilon\leq\frac{3}{9}, \\
    -(2\overline{a}(\mathbf{\Delta}_3^*)+\overline{b}(\mathbf{\Delta}_3^*))\varepsilon\leq\frac{3}{9},\\
 -(\overline{x}(\mathbf{\Delta}_3^*)+2\overline{y}(\mathbf{\Delta}_3^*))\varepsilon\leq\frac{3}{9}, \\
 -(\overline{x}(\mathbf{\Delta}_3^*)+2\overline{b}(\mathbf{\Delta}_3^*))\varepsilon\leq\frac{3}{9}, \\
 -(\overline{x}(\mathbf{\Delta}_3^*)+\overline{b}(\mathbf{\Delta}_3^*)
 +\overline{y}(\mathbf{\Delta}_3^*))\varepsilon\leq\frac{3}{9}. \\
   \end{array}
\right.
$$
It then follows that, for this case, $\varepsilon>0$ is $g^{(\ell)}_s$-valid if and only if $\varepsilon\leq\frac{3}{9\max\{3,12\}}=\frac{1}{36}$.

\item If $g^{(\ell)}_s(\mathcal C^*_3)=\frac{i}{9}$, where $i \geq 4$, then, similarly as above, $\varepsilon$ is $g^{(\ell)}_s$-valid if and only if $\varepsilon>0$ satisfies the following systems of inequalities:
$$
\left\{
  \begin{array}{ll}
    d_1\varepsilon\leq\frac{i}{9},\\
    d_2\varepsilon\leq\frac{i}{9},\\
    \cdots\\
    d_r\varepsilon\leq\frac{i}{9}, \\
 \end{array}
\right.
$$
for some integer $r$. It is easy to see that for all $1\leq j \leq r$, $d_j\leq-(\overline{x}(\mathbf{\Delta}_3^*)+2\overline{y}(\mathbf{\Delta}_3^*))=12$. So, for this case, we have $\varepsilon\leq\frac{1}{36}$ is $g^{(\ell)}_s$-valid.
\end{itemize}
Combining all the discussions as above, we conclude that $\varepsilon>0$ is $g_{\mathcal{S}_k}$-valid if and only if
$\varepsilon\leq\frac{1}{36}$, which completes the proof.

\section{Proof of Lemma~\ref{valid-k-4}} \label{proof-valid-k-4}

Firstly, note that for any $4$-sample $s$ and any $1\leq\ell\leq4$, \begin{equation}\label{g^l_s-4}
g^{(\ell)}_s(\mathcal C^*_4)=\frac{4}{16}m_{Ind_s}(\ell)-\frac{1}{16}\alpha(s),
\end{equation}
whence we have
$$
\{g^{(\ell)}_s(\mathcal C^*_4)\|s \in \mathcal S_4, 1\leq\ell\leq4\} \subseteq \left\{\frac{-3}{16},\frac{-2}{16}, \dots, \frac{15}{16},1 \right\}.
$$
We now consider the following cases:
\begin{itemize}
  \item If $g^{(\ell)}_s(\mathcal C^*_4)=\frac{-3}{16}$, which implies $m_{Ind_s}(\ell)=0$ and $\alpha(s)=3$ (see Equation (\ref{g^l_s-4})), then $\varepsilon$ is $g^{(\ell)}_s$-valid if and only if
$$
\left\{
  \begin{array}{ll}
    3\overline{y}(\mathbf{\Delta}_4^*)\varepsilon\leq\frac{3}{16},\\
    3\overline{b}(\mathbf{\Delta}_4^*)\varepsilon\leq\frac{3}{16},\\
    (2\overline{y}(\mathbf{\Delta}_4^*)+\overline{b}(\mathbf{\Delta}_4^*))\varepsilon\leq\frac{3}{16},\\
    (\overline{y}(\mathbf{\Delta}_4^*)+2\overline{b}(\mathbf{\Delta}_4^*))\varepsilon\leq\frac{3}{16}.
  \end{array}
\right.
$$
Noting from Example~\ref{perturbation-base-k} that $\overline{x}(\mathbf{\Delta}_3)=3$, $\overline{y}(\mathbf{\Delta}_4^*)=-5$, $\overline{a}(\mathbf{\Delta}_4^*)=-1$ and $\overline{b}(\mathbf{\Delta}_4^*)=3$, we have that, for this case, $\varepsilon>0$ is $g^{(\ell)}_s$-valid if and only if $3\overline{b}(\mathbf{\Delta}_4^*)\varepsilon\leq\frac{3}{16}$, i.e., $\varepsilon\leq\frac{1}{48}$.

  \item If $g^{(\ell)}_s(\mathcal C^*_4)=\frac{-2}{16}$, which implies $m_{Ind_s}(\ell)=0$ and $\alpha(s)=2$, then $\varepsilon$ is $g^{(\ell)}_s$-valid if and only if
$$
\left\{
  \begin{array}{ll}
    2\overline{y}(\mathbf{\Delta}_4^*)\varepsilon\leq\frac{2}{16},\\
    2\overline{b}(\mathbf{\Delta}_4^*)\varepsilon\leq\frac{2}{16},\\
    \overline{y}(\mathbf{\Delta}_4^*)\varepsilon+\overline{b}(\mathbf{\Delta}_4^*)\varepsilon\leq\frac{2}{16}.
   \end{array}
\right.
$$
Similarly as above, we deduce that, for this case, $\varepsilon>0$ is $g^{(\ell)}_s$-valid if and only $2\overline{b}(\mathbf{\Delta}_4^*)\varepsilon\leq\frac{2}{16}$, i.e., $\varepsilon\leq\frac{1}{48}$.

\item If $g^{(\ell)}_s(\mathcal C^*_4)=\frac{-1}{16}$, which implies $m_{Ind_s}(\ell)=0$ and $\alpha(s)=1$, then $\varepsilon$ is $g^{(\ell)}_s$-valid if and only if $\overline{b}(\mathbf{\Delta}_4^*)\varepsilon\leq\frac{1}{16}$, i.e., $\varepsilon\leq\frac{1}{48}$.

\item If $g^{(\ell)}_s(\mathcal C^*_4)=0$, then by definition, any $\varepsilon>0$ is $g^{(\ell)}_s$-valid.

\item If $g^{(\ell)}_s(\mathcal C^*_4)=\frac{1}{16}$, which implies $m_{Ind_s}(\ell)=1$ and $\alpha(s)=3$, then $\varepsilon$ is $g^{(\ell)}_s$-valid if and only if
$$
\left\{
  \begin{array}{ll}
    -(\overline{a}(\mathbf{\Delta}_4^*)+2\overline{y}(\mathbf{\Delta}_4^*))\varepsilon\leq\frac{1}{16}, \\
    -(\overline{a}(\mathbf{\Delta}_4^*)+2\overline{b}(\mathbf{\Delta}_4^*))\varepsilon\leq\frac{1}{16}, \\
    -(\overline{a}(\mathbf{\Delta}_4^*)+\overline{b}(\mathbf{\Delta}_4^*)
    +\overline{y}(\mathbf{\Delta}_4^*))\varepsilon\leq\frac{1}{16}.
   \end{array}
\right.
$$
We then deduce that, for this case, $\varepsilon>0$ is $g^{(\ell)}_s$-valid if and only if $-(\overline{a}(\mathbf{\Delta}_4^*)+2\overline{y}(\mathbf{\Delta}_4^*))\varepsilon\leq\frac{1}{16}$, i.e., $\varepsilon\leq\frac{1}{176}$.

\item If $g^{(\ell)}_s(\mathcal C^*_4)=\frac{2}{16}$, which implies $m_{Ind_s}(\ell)=1$ and $\alpha(s)=2$, then $\varepsilon$ is $g^{(\ell)}_s$-valid if and only if
$$
\left\{
  \begin{array}{ll}
    -(\overline{a}(\mathbf{\Delta}_4^*)+\overline{y}(\mathbf{\Delta}_4^*))\varepsilon\leq\frac{2}{16}, \\
    -(\overline{a}(\mathbf{\Delta}_4^*)+\overline{b}(\mathbf{\Delta}_4^*))\varepsilon\leq\frac{2}{16}.
  \end{array}
\right.
$$
We then infer than $\varepsilon>0$ is $g^{(\ell)}_s$-valid if and only if $-(\overline{a}(\mathbf{\Delta}_4^*)+\overline{y}(\mathbf{\Delta}_4^*))\varepsilon\leq\frac{2}{16}$, i.e., $\varepsilon\leq\frac{1}{48}$.

\item If $g^{(\ell)}_s(\mathcal C^*_4)=\frac{3}{16}$, which implies $m_{Ind_s}(\ell)=1$ and $\alpha(s)=1$, then $\varepsilon$ is $g^{(\ell)}_s$-valid if and only if $-\overline{a}(\mathbf{\Delta}_4^*)\varepsilon\leq\frac{3}{16}$, i.e., any $\varepsilon>0$ is valid.

\item If $g^{(\ell)}_s(\mathcal C^*_4)=\frac{i}{16}$, where $i=4, 5, \dots, 16$, then similarly as above, $\varepsilon$ is $g^{(\ell)}_s$-valid if and only if $\varepsilon>0$ satisfies the following system of inequalities,
$$
\left\{
  \begin{array}{ll}
    d_1\varepsilon\leq\frac{i}{16},\\
    d_2\varepsilon\leq\frac{i}{16},\\
    \cdots\\
    d_r\varepsilon\leq\frac{i}{16},\\
 \end{array}
\right.
$$
for some integer $r$. It is easy to see that for all $1\leq j \leq r$, $d_j \leq -(\overline{a}(\mathbf{\Delta}_4^*)+3\overline{y}(\mathbf{\Delta}_4^*))=16$. It then follows that, for this case, $\varepsilon\leq\frac{1}{64}$ is $g^{(\ell)}_s$-valid.
\end{itemize}
Combining all the discussions above, we conclude that $\varepsilon>0$ is $g_{\mathcal{S}_k}$-valid if and only if $\varepsilon \leq \frac{1}{176}$, which completes the proof.

\section{Proof of Lemma~\ref{balance-k-4}} \label{proof-balance-k-4}

For any $\mathcal C\in N(\mathcal C^{**}_4, \varepsilon)$, we write
$$
\mathbf{\Delta}=\mathcal C-\mathcal C^{**}_4 = \left(\left( \delta^{(1)}_{i, j}\right), \left( \delta^{(2)}_{i, j}\right), \left( \delta^{(3)}_{i, j}\right), \left( \delta^{(4)}_{i, j}\right)\right), \quad h_s(\mathbf{\Delta}) = g_s(\mathcal C)-g_s(\mathcal C^{**}_4).
$$
By Lemmas~\ref{sum-H_k(a,a,a)} and~\ref{sum-H_k(a,a,0)}, we have
$$
\sum_{s\in \mathcal{S}_4(3,3,0)} h_s(\mathbf{\Delta})=-16\sum_{i=1}^4\delta^{(i)}_{i,i}, \quad \sum_{s\in \mathcal{S}_4(3,3,3)} h_s(\mathbf{\Delta})=3\sum_{i=1}^4\delta^{(i)}_{i,i}+\sum_{i\neq j}\delta^{(j)}_{i,i}.
$$

We then have the following cases:

$\blacktriangleright$ For the samples in $\mathcal{S}_4(3,3,1)$, we write

\begin{equation}\label{sum-S_4(3,3,1)}
\sum_{s\in \mathcal{S}_4(3,3,1)} h_s(\mathbf{\Delta})=\sum_{\ell=1}^4\sum_{i=1}^4\sum_{j=1}^4 k^{(\ell)}_{i,j} \delta^{(\ell)}_{i,j},
\end{equation}
where the coefficients $k^{(\ell)}_{i,j}$ can be determined as follows. First, we consider $(i,i) \in [4]\times[4]$; and for simplicity only, assume $i=1$. There are 12 samples from $\mathcal{S}_4(3,3,1)$ containing $(1,1)$, more precisely, those samples from
$$
\{\{(1,1),(i,2), (j,3)\}\|i=1,3; j=1,2\} \cup \{\{(1,1),(i,2), (j,4)\}\|i=1,4; j=1,2\}
$$
$$
\cup \{\{(1,1),(i,3), (j,4)\}\|i=1,4; j=1,3\}.
$$
By (\ref{canonical-signal-property}), it is easy to see that $k^{(1)}_{1,1}=12$ and  $k^{(2)}_{1,1}=k^{(3)}_{1,1}=k^{(4)}_{1,1}=4$. The other coefficients of $\delta^{(\ell)}_{i,i}$ can be obtained similarly as
\begin{itemize}
  \item $k^{(i)}_{i,i}=12$, $1\leq i\leq 4$;
  \item $k^{(j)}_{i,i}=4$, $i\neq j$.
\end{itemize}
We now consider $(i,j)\in [4]\times[4]$ for $i\neq j$; and for simplicity only, assume $(i, j)= (2,1)$. There are 8 samples of $\mathcal{S}_4(3,3,1)$ containing $(2,1)$, more precisely, those samples from
$$
\{\{(2,1),(2,2),(i,3)\}\|i=1,2\} \cup \{\{(2,1),(i,2),(3,3)\}\|i=1,3\}
$$
$$
\cup \{\{(2,1),(2,2),(i,4)\}\|i=1,2\} \cup \{\{(2,1),(i,2),(4,4)\}\|i=1,4\}.
$$
By (\ref{canonical-signal-property}), it is easy to see that $k^{(1)}_{2,1}=k^{(2)}_{2,1}=8$ and  $k^{(3)}_{2,1}=k^{(4)}_{2,1}=0$. The other coefficients of $\delta^{(\ell)}_{i,j}$ can be obtained similarly as
\begin{itemize}
  \item $k^{(i)}_{i,j}=k^{(i)}_{j,i}=8$, $1\leq i,j\leq 4$;
  \item $k^{(\ell)}_{i,j}=0$, if $i,j,\ell$ are distinct.
\end{itemize}
Finally, by (\ref{sum-column-row-k}), we have
$$
\sum_{s\in \mathcal{S}_4(3,3,1)} h_s(\mathbf{\Delta})=-4\sum_{i=1}^4\delta^{(i)}_{i,i}+4\sum_{i\neq j}\delta^{(j)}_{i,i}.
$$

$\blacktriangleright$ For the $81$ samples in $\mathcal{S}_4(4,4,0)$, as in the previous case, we write
\begin{equation}\label{sum-S_4(4,4,0)}
\sum_{s\in \mathcal{S}_4(4,4,0)} h_s(\mathbf{\Delta})=\sum_{\ell=1}^4\sum_{i=1}^4\sum_{j=1}^4 k^{(\ell)}_{i,j} \delta^{(\ell)}_{i,j}.
\end{equation}
Note that for $(i,i) \in [4] \times [4]$, since there is no sample containing $(i,i)$, $k^{(\ell)}_{i,i}=0$. We then consider $(i,j)\in [4]\times[4]$ for $i\neq j$; and for simplicity only, assume $(i, j)= (2,1)$. There are 27 samples of $\mathcal{S}_4(4,4,0)$ containing $(2,1)$, more precisely, those samples from
\begin{equation*}
\begin{split}
\{\{(2,1),(i,2), (j,3),(\ell,4)\}\|i\neq 2,j\neq 3,\ell\neq 4; 1\leq i,j,\ell \leq 4\}.
\end{split}
\end{equation*}
It is easy to verify that  $k^{(1)}_{2,1}=k^{(2)}_{2,1}=k^{(3)}_{2,1}=k^{(4)}_{2,1}=27$, and the other coefficients of $\delta^{(\ell)}_{i,j}$ can be obtained similarly. All in all, we have $k^{(\ell)}_{i,j}=27$, for $i\neq j$. Hence, we have
$$
\sum_{s\in \mathcal{S}_4(4,4,0)} h_s(\mathbf{\Delta})=-27\sum_{i=1}^4\sum_{j=1}^4\delta^{(i)}_{j,j}.
$$

Combining the above results, we have
\begin{equation}\label{sum-S_4(C**)}
\frac{1}{8}\sum_{s\in \mathcal{S}_3(3,3,0)} h_s(\mathbf{\Delta})+\frac{1}{4}\sum_{s\in \mathcal{S}_4(3,3,1)} h_s(\mathbf{\Delta})+2\sum_{s\in \mathcal{S}_4(3,3,3)} h_s(\mathbf{\Delta})+\frac{1}{9}\sum_{s\in \mathcal{S}_4(4,4,0)} h_s(\mathbf{\Delta})=0,
\end{equation}
which completes the proof.

\end{document}